\documentclass[a4paper,10pt]{article}
\usepackage[utf8x]{inputenc}
\usepackage{amsmath,amssymb,amsthm}
\usepackage{color}
\usepackage{arydshln}
\usepackage{bm}
\usepackage{textcomp}
\numberwithin{equation}{section}
\newtheorem{thm}{Theorem}[section]
\newtheorem{cor}{Corollary}
\newtheorem{lem}{Lemma}[section]
\theoremstyle{definition}

\newtheorem{definition}{Definition}[section]
\newtheorem{example}{Example}[section]
\usepackage{enumerate}

\def\showcomment{}

\ifdefined\showcomment

  \newcommand{\comm}[1]{\emph{\textcolor{red}{\textbf{SC:}#1}}}

\else
  \newcommand{\comm}[1]{} 

\fi

\definecolor{nblue}{rgb}{0,0,0} 
\definecolor{nbrown}{rgb}{0,0,0}
\definecolor{mbrown}{rgb}{0,0,0}
\definecolor{dblue}{rgb}{0,0,0} 
\definecolor{myblue}{rgb}{0,0,0}
\definecolor{mygreen}{rgb}{0,0,0}
\definecolor{purple}{rgb}{0,0,0}

\title{{\color{dblue}Zero-error {\color{mbrown}Slepian-Wolf} Coding of {\color{purple}Confined-} Correlated Sources with Deviation Symmetry}}

\author{Rick Ma\thanks{R. Ma was with the Department of Mathematics at the Hong Kong University of Science and Technology, Hong
Kong.} and Samuel Cheng\thanks{S. Cheng is with the School
of Electrical and Computer Engineering, University of Oklahoma, Tulsa,
OK, 74135 USA email: samuel.cheng@ou.edu. 
This work was supported in part by NSF under grant CCF 1117886.
 }
}

\begin{document}

\maketitle

\begin{abstract}
{\color{black} In this paper, we {\color{nbrown} use linear codes to study zero-error} {\color{mbrown}Slepian-Wolf} coding of a set of {\color{dblue} sources with deviation symmetry, where the sources}  are generalization of the Hamming sources \color{myblue}over {\color{nbrown}an} arbitrary field\color{black}. We extend \color{myblue} 
our previous codes, \color{black} {\em Generalized Hamming Codes for Multiple Sources}, to {\em Matrix Partition Codes} and use the latter to efficiently compress {\color{dblue} the target sources}. We further show that  
\color{myblue}every perfect or linear-optimal code is a Matrix Partition Code}. \color {black}
 We also present some conditions when {\color{dblue} Matrix Partition} Codes are perfect and/or linear-optimal. Detail discussions of {\color{dblue} Matrix Partition} Codes on Hamming sources are given at last as examples. 

\end{abstract}

\section{Introduction}


Slepian-Wolf (SW) Coding or Slepian-Wolf problem refers to 
separate encoding of multiple correlated sources but joint lossless decoding of the sources \cite{SlepianW:73}. 
Since then, many researchers have looked into ways to implement SW coding efficiently. Noticeably, 
Wyner was the first who realized that 
linear coset codes can be used to tackle the problem
\cite{Wyner:74}. Essentially, considering the source from each terminal as a column vector, the encoding output will simply be the multiple of a ``fat''\footnote{The input matrix is ``fat'' so that the length of encoded vector will be shorter than {\color{dblue} or equal to} the original.} coding matrix and the input vector. The approach was 
popularized by Pradhan {\em et al.} more than two decades later \cite{PradhanR:99}.
Practical syndrome-based schemes for S-W coding using channel codes 
have been further studied in
\cite{StankovicLX:06,pradhan2005gcc,YangCXZ:03,LiverisXG:03b,ChouPR:03,MitranB:02a,WangO:01,Servetto:02,zamani2009flexible}. 

Unlike many prior works focusing on near-lossless compression, in this work we consider true lossless compression (zero-error reconstruction) in which sources are {\em always} recovered losslessly 
{\color{black}
\cite{al1997zero,koulgi2001minimum,koulgi2003zero,yan2000instantaneous,ma2011universality}.}
\color{myblue}
 So we say the SW code can {\em compress} $S$ only if any source tuple in $S$ can be reconstructed losslessly. \color{black}
Obviously, a SW code can compress $S$ 
if and only if its encoding map restricted to $S$ is injective (or 1-1). 

{\color{mbrown} The source model for zero-error SW coding can be quite a bit different from the typical probabilistic model studied in classic SW coding literatures.}
{\color{mbrown} For example, for} highly correlated sources, we expect that sources from most terminals are likely to be the same. The trivial case is when all $s$ sources are identical. 
{\color{black} The next (simplest non-trivial) possible case is when all sources except one are identical, and in the source that is different from the rest, only one bit differs from the corresponding bit of other sources. } Such source is known to be {\em Hamming source} \cite{ChengM:DCC2010} since it turns out that it is closely related to Hamming codes.


In \cite{ChengM:DCC2010},
we described
a generalized syndrome based coset code and extended 
the notions of a packing bound and a perfect code 
from regular channel coding to SW coding with an arbitrary number of sources. 
In \cite{MaC:rick2},
we introduced the notion of Hamming Code for Multiple Sources (HCMSs) as a perfect code solution for Hamming sources. {\color{black} Moreover, we have shown that {\color{black} there exist an infinite number of HCMSs} for three sources. However, we have also pointed out that not all perfect codes for Hamming sources can be represented as HCMSs.} In \cite{ma2011universality}, we extended HCMS to {\em generalized HCMS}. And we showed that 
{\color{black} any perfect SW code {\color{black} for} a Hamming source is equivalent to a generalized HCMS (c.f. Theorem 3 in \cite{ma2011universality})}.   
%

{\color{black}
Despite our prior results, Hamming source is a very restricted kind of sources and only binary Hamming sources had been studied in the past. 
In this paper, we extend our prior works to input sources in {\color{dblue} arbitrary} fields. Moreover, we introduce a much general kind of sources 
{\color{dblue} with deviation symmetry as to be spelled out in Definition \ref{def:szeto}. We will also show
such sources}
can be handled systematically with 
{\color{dblue}
the proposed {\em Matrix Partition Codes}, which 
 can be interpreted as an extension of the generalized HCMS described in \cite{ma2011universality}.}
We also show that the Matrix Partition Codes {\color{nblue} of} any linear-optimal compression (i.e., higher compression is impossible) is  
{\color{dblue} a Matrix Partition Code}.
We also present some conditions when {\color{dblue} the Matrix Partition} Codes are perfect and/or linear-optimal. Some more detail discussions are further given for special cases such as Hamming sources. 

Let us briefly summarize here the notations and conventions used in this paper. Matrices are generally denoted with upper case letters while vectors are denoted by lower case letters. Sets are denoted using script font and fields are denoted using blackboard bold letter font. As matrices are used as mapping during encoding, we call a matrix {\em injective} ({\em surjective}) when the mapping corresponding to the matrix is injective (surjective). We may also specify the domain of the mapping. When none is specified, it is understood that the domain contains all possible vectors. For example, we say $A |_{\mathcal{S}}$ is injective if the mapping corresponding to matrix $A$ with input domain $\mathcal S$ is injective. In other words, for any $\sigma_1, \sigma_2 \in \mathcal S$ and $\sigma_1 \neq \sigma_2$, $A \sigma_1 \neq A \sigma_2$. Further, we call a matrix $A$ a {\em row basis matrix} of a matrix $B$ if rows of $A$ form a basis of the row space of $B$. 
}

This rest of the paper is organized as follows. In the next section, we will introduce the {\color{dblue} target sources with {\color{mygreen} deviation symmetry} and the Matrix Partition Codes}. We will {\color{nblue} include some motivations of the setup in Section \ref{sect:motivation} and} {\color{dblue} also derive} the maximum possible compression {\color{dblue} achievable by a Matrix Partition} Code. In Section \ref{sect:perfect}, we discuss when a {\color{dblue} Matrix Partition} Code will be {\em perfect}. Moreover, we use generalized Hamming sources as an example and derive the necessary conditions of perfectness. In Section \ref{sect:universality}, we present the major result of this paper---the uniqueness of {\color{dblue} Matrix Partition} Codes. 
In Section \ref{sect:hamming}, we will present some new results that are restricted to a subset of {\color{dblue} sources with {\color{mygreen} deviation symmetry}}, the Hamming sources. 
\color{myblue}
In Section \ref{sect:fixing}, before concluding the paper, we will determine the condition on the source under which actual compression is possible. \color{black}

\section{{\color{dblue} Target Sources and Proposed Codes}}

{\color{nblue}
\subsection{{\color{mbrown}Confined-Correlated Source and Connection to Classic Probabilistic Source Models}}
\label{sect:motivation}

Slepian-Wolf (SW) coding {\color{mbrown}is typically referred to as} the (near-)lossless compression of jointly correlated sources with separate encoders but a joint decoder. 
And the source is {\color{mbrown}usually} modeled probabilistically in the {\em classic} setup. More precisely, an $s$ terminal system can have the sources $X_1$, $X_2$, $\cdots$, $X_s$, sampled from
some joint distribution $p(x_1,x_2,\cdots,x_s)$ at each time instance {\color{mbrown}independent of another time instances}. 

{\color{mbrown}Let us consider a simple example of a two terminal source ($s=2$) with $Pr(X_1=0,X_2=0)=Pr(X_1=1,X_2=1)=0.4$ and $Pr(X_1=1,X_2=0)=Pr(X_1=0,X_2=1)=0.1$. 
One can see that the marginal distributions of both terminals are uniformly distributed
(i.e., $Pr(X_1=1)=Pr(X_1=0)=Pr(X_2=1)=Pr(X_2=0)=0.5$). Thus, any sequence drawn from each terminal will be equally likely and applying {\em variable length} code on a sequence from one terminal will not improve compression efficiency. 
Note that this is true in many scenarios when such kind of symmetry exists and will be described more precisely in Definition \ref{def:szeto}.
}

Given a block of $n$ source {\color{mbrown}sequence} tuples $({\bf x}_1,\cdots, {\bf x}_s)$ sampled from the joint source (where each ${\bf x}_i$ has length $n$),
the encoders
$Enc_1, \cdots, Enc_s$ apply on their corresponding source {\color{mbrown}sequences} only. That is, we have the $i^{th}$ encoding output 
\begin{align}
{\bf y}_i=Enc_i({\bf x}_i) 
\end{align}
{\color{mbrown}has length $m_i$ and is} independent of ${\bf x}_j$, $j\neq i$. {\color{mbrown}For any encoder $i$, we choose to have the codeword length $m_i$ fixed (independent of the input sequence) since as we mentioned in the previous paragraph, variable length coding is not going to increase efficiency for the symmetric source that we consider here.}

Receiving the compressed outputs {\color{mbrown}${\bf y}_1,\cdots, {\bf y}_s$}, a joint decoding map $Dec$ will try to recover the source blocks from all terminals. That is, 
\begin{align}
 ({\bf \hat x}_1,\cdots,{\bf \hat x}_s) = Dec({\bf y}_1,\cdots, {\bf y}_s),
\end{align}
where ${\bf \hat x}_i$ is the estimate of ${\bf x}_i$. 

The problem of the {\color{mbrown}aforementioned} probabilistic setup is that {\color{mbrown}for a finite $n$,} true lossless compression is generally not possible.
Denote $\mathcal{S}$ as the set of all possible $({\bf x}_1, \cdots,{\bf x}_s)$ that can be sampled from the source. Define an encoding map 
\begin{align}
 Enc({\bf x}_1, \cdots,{\bf x}_s)=(Enc_1({\bf x}_1), \cdots, Enc_s({\bf x}_s)).
\end{align}
Obviously, for a SW coding scheme to be {\em truly} lossless, we must have the restricted map $Enc|\mathcal{S}$ to be injective {\color{mbrown}(i.e., no two possible inputs will result in the same encoded output)}. {\color{mbrown}Denote $\mathcal{X}_i$ and $\mathcal{Y}_i$ as the alphabets of input and output of the $i^{th}$ encoder, respectively.} For a finite $n$ and a general distribution {\color{mbrown} that $p(x_1,\cdots,x_s)\neq 0$ for any combination of scalars $x_1$,$\cdots$,$x_s$}, every $({\bf x}_1, \cdots,{\bf x}_s)$ will have a non-zero probability and thus is {\color{mbrown} in $\mathcal{S}$}. This essentially means that the size of the compressed source {\color{mbrown}$|\mathcal{Y}^{m_1}_1 \times \cdots \times \mathcal{Y}^{m_s}_s|$}, which has to be larger than $|\mathcal{S}|$ for true {\em lossless} {\color{mbrown}recovery}, is just $|\mathcal{X}^n_1 \times \cdots \times \mathcal{X}^n_s|$. 
{\color{mbrown} Therefore, one cannot have both true lossless recovery and real compression (i.e., the net encoded output is smaller than the input) in this case.} 

The catch here is that in a classic SW setup, we will allow $n$ to go to infinity. In consequence, for any distribution, we will have some $({\bf x}_1, \cdots,{\bf x}_s)$ (the jointly typical sequences) to have much higher probabilities than the rest, in such an extend that the other joint sequences will have negligible probabilities {\color{mbrown}(essentially 0 as $n$ goes to infinity)} and can be ignored. Thus we will have $|\mathcal{S}|$ smaller than $|\mathcal{X}^n_1 \times \cdots \times \mathcal{X}^n_s|$ if we exclude joint sequences that almost never happen. And this gives us a near-lossless compression {\color{mbrown}for a very large but finite $n$}. 

While the probabilistic approach of the classic SW setup leads to interesting theoretical performance bounds, an infinite $n$ is not realistic in practice. In particular, unlike a channel coding problem that is typically designed to operate at a very high sampling rate,  the sampling rate of a SW problem is not a design parameter but is determined by the nature of the source. For example, in a typical scenario where the sources are the temperature readings sampled from different locations, let say, every five minutes. A rather small $n$, say $100$, will already correspond to over eight hours of delay. Such a delay may not be acceptable in practical scenarios. 

To accommodate a finite delay, one may give up either the lossless requirement or the 
\color{purple}conventional \color{myblue}  probabilistic model. Giving up the lossless requirement will result in the general multiterminal source coding problem, a much more complicated setup where a general theoretical rate-distortion limit is still unknown. Instead, we will forfeit the conventional probabilistic model in this paper. 
Unlike the conventional SW case, the set $\mathcal{S}$ containing all possible joint sequences  is a proper subset of $ \mathcal{X}^n_1 \times \cdots \times \mathcal{X}^n_s$ even when $n$ is finite. We will call such source \color{purple} {\em confined-correlated source }\color{myblue}  to distinguish it from the conventional case. Another interpretation is simply that joint sequences are \color{purple} directly drawn from the set $\mathcal{S}$\color{myblue}.

\begin{definition}[\color{purple}Confined-\color{myblue}Correlated Source]
 Given a subset $\mathcal{S}$ of $\mathcal{X}^n_1 \times \cdots \times \mathcal{X}^n_s$, we call the source from which joint sequences are drawn a {\em confined-correlated source} if the probability of having a joint sequence \color{purple} outside $\mathcal {S}$ is zero.
{\color{mbrown} Note that once $\mathcal {S}$ has been defined, we will treat every element of its equally, regardless of the original probabilistic structure. And our discussion always directly starts with a given $\mathcal {S}$. Hence our coding is completely characterized by the source set $\mathcal{S}$.} 
\end{definition}

\color{myblue}
{\color{mbrown}Just as most other works in the literatures,} we will focus on linear {\color{mbrown}code} in this paper. 
{\color{mbrown}Here is our mathematical setting. Let $\mathbb F$ be an arbitrary field.
Let $s, n, m_1,\cdots, m_s$ be integers that $s\ge 2, n\ge 1, m_1\ge 0,\cdots, m_s\ge 0$. We restrict that $\mathcal{X}_1 =\cdots =\mathcal{X}_s=\mathcal{Y}_1 = \cdots = \mathcal{Y}_s = \mathbb{F}$}.
Hence the source $\mathcal S$ is a subset of 
$\overbrace{\mathbb{F}^n\times\cdots\times\mathbb{F}^n}^{s\mbox{ terms}}$ 
{\color{mbrown}and the output codeword space $\mathcal C=\mathbb F^{m_1}\times\cdots\times\mathbb F^{m_s}$}. Later we 
will further require $\mathcal S$ to be a source with deviation symmetry (Definition \ref{def:szeto}).
The encoding map is linear in the sense that  $Enc_i({\bf x}_i) = H_i {\bf x}_i$, where $H_i$ are $m_i\times n$ {\em{\color{mbrown} encoding} matrices} over $\mathbb F$ for all $i$. 
We can have lossless compression only if $(H_1,\cdots, H_s)|_\mathcal{S}$ is injective. Since we will only consider linear lossless compression in this paper, we will simply refer such kind of compression to as {\em compression} in the following, {\color{mbrown}except for emphasis.}





Beside injectivity, we certainly want the output space $\mathcal C$ to be small. 
For this purpose, we define some measure to quantify its size.
\begin{definition}[Total Code Length, {\color{mbrown}Compression Sum-Ratio},  {\color{mbrown}Compression Ratio} Tuples]
 The total code length $M$ of compression  $(H_1, \cdots, H_s)$ is given by $M=m_1+\cdots +m_s$, which is dim$\mathcal C$. We also define the compression {\color{mbrown}sum-ratio} as $M/n$ and the  {\color{mbrown}compression ratio} tuples as 
$(m_1/n, \cdots, m_s/n)$.
\end{definition}

\begin{definition}[Linear-Optimal Compression]
 A linear compression is said to be linear-optimal if there is no other linear compression with respect to the same source resulting in a shorter total code length. 
\label{defO}
\end{definition}

\begin{definition}[Perfect Compression]
The compression  $(H_1,H_2, \cdots, H_s)$ is said to be perfect if $\mathbb F$ is finite and 
$|\mathcal C|=|\mathbb F|^M=|\mathcal S|$.
\label{defP}
\end{definition}
While there is always a linear-optimal compression scheme, perfect compression does not always exist. A perfect compression scheme is obviously linear-optimal but the converse may not be true. 

}

\subsection{\color{dblue} {\color{nblue}{\color{mbrown}Confined-}Correlated} Sources with {\color{mygreen}Deviation Symmetry}}
%
%
%

{\color{nblue} Even with the restriction of linear codes, the considered problem is still too general.
We will introduce {\color{mbrown}the symmetric} constraint {\color{mbrown}mentioned in the last subsection} to our target source. Namely, if a joint sequence {\color{mbrown}$\sigma$} belongs to a source $\mathcal{S}$, so does a uniform shift {\color{mbrown}$\sigma + \bf (v,v,\cdots,v)$}. The condition is a rather mild one. Actually, if we imagine that each source output are just readings derived from a common base source, such symmetry will natural arise if the observers are symmetrically setup. 

Let us consider the equivalent relation on $\overset{\mbox{$s$ terms}}{\overbrace{\mathbb{F}^n \times \cdots \times \mathbb{F}^n}}$ 
that 
{\color{mbrown}
\begin{equation}
\sigma_1 \sim \sigma_2 \mbox{ iff $\sigma_1-\sigma_2= \bf (v,\cdots,v)$ for some ${\bf v}\in \mathbb F^n$. }
\label{eqn1.13}
\end{equation}}
Then the equivalence classes derived from the equivalence relation partition the joint sequence space $\mathbb{F}^n \times \cdots \times \mathbb{F}^n$ and we may redefine our target source as follows.

}

%
%

\begin{definition}[{\color{dblue} Sources with Deviation Symmetry}]
\label{def:szeto}
 A {\color{mbrown}confined-}correlated source $\mathcal S \subset \overbrace{\mathbb F^n \times \cdots \times \mathbb  F^n}^{\mbox{$s$ terms}}$ 
is said to be a {\color{dblue} {\em source with {\color{mygreen} deviation symmetry}}}
{\color{nblue}
if $\mathcal S$ is an union of equivalence classes derived from the equivalence relation specified by \eqref{eqn1.13}. 
}
\end{definition}

{\color{nblue}
A source with deviation symmetry is completely characterized by its composite equivalence classes, where each can in term be specified by any one element of the equivalence class. Let us define a representative set as follow. 

\begin{definition}[Representative Set]
 A representative set $\mathcal{D}$ of a source with deviation symmetry $\mathcal{S}$ contains exactly one element of each equivalence class that is a subset of $\mathcal{S}$. 
 \label{def:representative_set}
\end{definition}

Obviously, $\mathcal{D}$ is not unique and
since $\mathcal{D}$ contains exactly one element from an equivalence class, we have the following property.
\begin{equation}
\delta \in \mathcal D 
\Rightarrow {\bf
(v,\cdots,v)+\delta \notin \mathcal D,  \forall \mbox{ non-zero } v}\in \mathbb
F^n.
\label{eqn1.3}
\end{equation}
{\color{nblue} Furthermore, if both $\mathcal D$ and $\mathcal E$ are representation sets of $\mathcal S$,} then there exists a mapping
${\bf v}:\mathcal D\mapsto \mathbb F^n$ such that
\begin{equation}
\mathcal E=\{\bf ({\bf v}(\delta),\cdots,{\bf v}(\delta))+\delta\hspace{0.1cm} | \hspace{0.1cm}\delta \in \mathcal D\}.
\label{eqn1.12}
 \end{equation}
And from Definition \ref{def:representative_set},  source $\mathcal{S}$ can be completely characterized by $\mathcal{D}$ with  
\begin{equation}
\mathcal S(\mathcal D)=\{{\bf (v,\cdots,v)+\delta \hspace{0.1cm} | \hspace{0.1cm}v}\in \mathbb F^n, \delta\in \mathcal D\}.
\label{eqn1.1} 
\end{equation}
Moreover, 
$\forall \sigma \in \mathcal S$, $\exists$ a unique $\delta \in \mathcal D$ and
a unique ${\bf v}\in \mathbb F^n$ such that 
\begin{equation}
\sigma =\bf (v,\cdots,v) + \delta. 
 \label{eqn1.11}
\end{equation}
Indeed, \eqref{eqn1.1} guarantees the existence of such $\bf v$ and $\delta$ and if 
$\bf (v,\cdots,v)+\delta=(u,\cdots,u)+\zeta$ with
another pair of ${\bf u}\in \mathbb F^n, \zeta\in\mathcal D$, then both $\delta$ and 
$\bf (v-u,\cdots,v-u)+\delta$ are in $\mathcal D$. By \eqref{eqn1.3}, we will
have $\bf v=u$ and hence $\delta=\zeta$, and this guarantees the uniqueness. 
Finally, if $\mathbb{F}$ is a finite set, it is easy to see that 
\begin{equation}
|\mathcal S|=|\mathbb F|^n |\mathcal D|.                                        
\label{eqn1.4}
\end{equation}
}
%
%

%


{\color{nblue}
\begin{example}[Hamming Sources]
 A Hamming source  \cite{ma2011universality} $\mathcal S$ as defined by
\begin{equation}
\mathcal S=\{{\bf(v,...,v)}+(\underset{i \mbox{ terms}}{\underbrace{{\bf 0},\cdots,a
{\bf e}_j}},\cdots,{\bf 0}) | a\in \mathbb F, 1 \le i \le s, 1 \le j \le n\}
\label{eqn1.5} 
\end{equation}
is clearly a source with deviation symmetry, where ${\bf e}_j$ is a length-$n$ vector with zeros for all but the $j^{th}$ component being 1. For $s \ge 3$, we can simply choose the representative set as 
\begin{equation}
\mathcal D=\{(\underset{i \mbox{ terms}}{\underbrace{{\bf 0},\cdots,a {\bf e}_j}},\cdots,{\bf 0})
| a\in \mathbb F,  1 \le i \le s, 1 \le j \le n\}
\label{eqn1.6}
\end{equation}
and we have
\begin{equation}
|\mathcal S|=|\mathbb F|^n(1+s(|\mathbb F|-1)n) \mbox{ for finite } \mathbb F. 
\label{eqn1.7}
 \end{equation}
But when $s=2$, \eqref{eqn1.6} is not a good choice as
\begin{equation}
{\bf (0, e}_1)=({\bf e}_1,  {\bf e}_1)+(-{\bf e}_1, {\bf 0}),
\label{eqn1.8} 
\end{equation}
which contravenes the restriction \eqref{eqn1.3}. Instead, we {\color{dblue} may choose}
\begin{equation}
\mathcal D=\{({\bf 0},a {\bf e}_j) | a\in \mathbb F, 1\le j\le n\} 
\label{eqn1.9}
\end{equation}
and get
\begin{equation}
|\mathcal S|=|\mathbb F|^n(1+(|\mathbb F|-1)n) \mbox{ for finite } \mathbb F.   
\label{eqn1.10}
\end{equation}

\end{example}
}



{\color{nblue} Before we end this section, we would like to define a vectorized correspondence of $\mathcal D$ for later usage.} Let
\begin{equation}
\color{myblue}\tilde{\mathcal D}\color{black}=\left\{\left.\begin{pmatrix}
{\bf d}_1 \\
            {\bf d}_2 \\
             \vdots \\
             {\bf d}_s      
     \end{pmatrix}
   \right| ({\bf d}_1,\cdots,{\bf d}_s)\in \mathcal D \right\}\subset \mathbb F^{sn}. 
 \label{eqn2.1}
\end{equation}
We have 
\begin{equation}
|\mathcal D| = |\color{myblue}\tilde{\mathcal D}\color{black}| \mbox{ for finite field $\mathbb F$}.
\label{eqn1.15}
\end{equation}

\subsection{{\color{dblue} Pre-Matrix Partition} Codes}



{\color{nblue}

}
The following theorem suggests a way to construct codes for {\color{dblue} sources with {\color{mygreen} deviation symmetry}}. We call such codes {\color{dblue} Pre-Matrix Partition} Codes as the name {\em \color{dblue} Matrix Partition Codes} will be reserved to the more refined codes to be discussed shortly afterward. 

\begin{thm}[{\color{dblue} Pre-Matrix Partition} Codes]
\label{thm2.1}
Let $P$ be an $r \times sn$ matrix $(r\in \mathbb Z_+)$ over $\mathbb F$ s.t.
\begin{equation}
                             P|_{\color{myblue}\tilde{\mathcal D}\color{black}} \mbox{ is } 
                             injective.        
\label{eqn2.2}
\end{equation}
Suppose P can be partitioned into
\begin{equation}
 P=[Q_1| \cdots      |Q_s] \mbox{ s.t. } Q_1+\cdots +Q_s=0, 
\label{eqn2.3}
\end{equation}
where all $Q_i$ are $r \times n$ matrices.
Then for any matrix $T'$ that
\begin{equation}
\begin{pmatrix}
Q_1 \\
 \vdots \\
 Q_{s}\\
  T'
 \end{pmatrix}
      \mbox{ forms an injective matrix},                                     
\label{eqn2.5}
 \end{equation}
 we let $\{G'_i | 1 \le i \le s\}$ be a row partition of $T'$, ie
\begin{equation}
\begin{pmatrix}
G'_1 \\
 \vdots \\
 G'_{s}\\
   \end{pmatrix}             = T'.
\label{eqn2.6} 
\end{equation}
Encoding matrices $(H_1,...,H_s)$ with
\begin{equation}
 \mbox{null} H_i=\mbox{null} \begin{pmatrix}
                G'_i \\Q_i
               \end{pmatrix}
\mbox{ for all $i$}   
\label{eqn2.7}                                              
\end{equation}
form a compression that we name Pre-Matrix Partition Code.
\end{thm}

\begin{proof}
Define {\color{mbrown}$\mathcal S_+ = \{ \sigma_1 - \sigma_2 | \sigma_1, \sigma_2\in \mathcal S \}$}.
Since $(H_1,...,H_s) |_\mathcal S$ is injective iff 
\begin{equation}
\mbox{null} H_1 \times \mbox{null} H_2  \times \cdots  \times \mbox{null} H_s \cap \mathcal S_+ =
\{0\},               
\label{eqn2.7b}
 \end{equation}
 {\color{black} the validity of the compression solely depends on the null spaces of coding matrices $H_i$}. Thus we only need to prove for the special case when 
$ H_i=\begin{pmatrix}
       G'_i \\Q_i
      \end{pmatrix}
$ for all $i$.  

Suppose
\begin{equation}
\begin{pmatrix}
 G'_i \\Q_i
\end{pmatrix}
({\bf u}+{\bf d}_i) = 
\begin{pmatrix}
 G'_i \\Q_i
\end{pmatrix}
({\bf v}+{\bf f}_i) 
\mbox{ for $1 \le i \le s$},           
\label{eqn2.11}
\end{equation}
where ${\bf u},{\bf v} \in \mathbb F^n$; $({\bf d}_1,\cdots,{\bf d}_s), ({\bf f}_1,\cdots,{\bf f}_s)\in \mathcal D$.
We get
\begin{equation}
\begin{pmatrix}
 G'_i \\Q_i
\end{pmatrix}
 ({\bf w}+{\bf d}_i-{\bf f}_i)={\bf 0} 
\mbox{ for $1 \le i \le s$, }
\label{eqn2.12} 
\end{equation}
where $\bf w=u-v$. In particular,
\begin{equation}
Q_i({\bf w}+{\bf d}_i-{\bf f}_i)={\bf 0} \mbox{ for $1 \le i \le s$ }       
\label{eqn2.13}
%
\end{equation}
and hence
\begin{equation}
Q_1({\bf w}+{\bf d}_1-{\bf f}_1)+\cdots +Q_s({\bf w}+{\bf d}_s-{\bf f}_s)={\bf 0}.
\label{eqn2.14}
%
\end{equation}
By \eqref{eqn2.3}, we get
\begin{align}
Q_1({\bf d}_1)+...+Q_s({\bf d}_s)&=Q_1({\bf f}_1)+..+Q_s({\bf f}_s) \\
\Rightarrow
P \begin{pmatrix}
   {\bf d}_1 \\ \vdots \\ {\bf d}_s
  \end{pmatrix}
= 
P \begin{pmatrix}
   {\bf f}_1 \\ \vdots \\{\bf f}_s
  \end{pmatrix}.
\label{eqn2.15}
%
\end{align}
By \eqref{eqn2.2},
\begin{equation}
\begin{pmatrix}
 {\bf d}_1 \\ \vdots \\ {\bf d}_s
\end{pmatrix}
= 
\begin{pmatrix}
 {\bf f}_1 \\ \vdots \\ {\bf f}_s
\end{pmatrix}.
\label{eqn2.16}
\end{equation}
Then \eqref{eqn2.12} become
\begin{equation}
\begin{pmatrix}
G'_i \\
 Q_i
 \end{pmatrix}
 ({\bf w}) = {\bf 0} \mbox{ for $1 \le i \le s$},
\label{eqn2.17}
\end{equation}
 which gives                       
\begin{equation}
\begin{pmatrix}
 Q_1 \\ \vdots  \\ Q_{s} \\ T' 
\end{pmatrix}
({\bf w}) = {\bf 0} \mbox{ {\color{black} (c.f. \eqref{eqn2.6})}.}
 \label{eqn2.18}
\end{equation}
Hence we must have $\bf w=u-v=0$ by 
\eqref{eqn2.5}.

\end{proof}

 \color{myblue}
The Pre-Matrix Partition Codes fulfill the basic requirement of our definition of compression, {\color{mbrown}i.e.,} injectivity. 
They do not take the sizes of the output codeword spaces into account.
 In the following, we are going to put more restriction on the codes to maximize the compression efficiency {\color{mbrown}(in the sense of Theorem \ref{thm2.2})}. 
\color{myblue}

\color{myblue} 

\subsection{\color{dblue} Matrix Partition Codes}
\label{subsection2.3}

\color{myblue}

\begin{definition}[Matrix Partition Codes]
Let $P$ be a matrix satisfying \eqref{eqn2.2} and \eqref{eqn2.3}.
Let 
{\color{nblue}
\begin{equation}
Y \mbox{ be a row basis matrix of } 
\begin{pmatrix}
Q_1 \\
\vdots \\
Q_s
\end{pmatrix}
\label{eqn2.18g} 
\end{equation}}
and $T$ be a matrix s.t. 
\begin{equation}
\begin{pmatrix}
 Y \\T
\end{pmatrix}
 \mbox{ is an invertible $n \times n$ matrix.                } 
 \label{eqn2.18z}
\end{equation}
Then we call the compression with encoding matrices
\begin{equation}
\color{myblue}{U_1}\color{black} \begin{pmatrix}
  C_1 \\ G_1
 \end{pmatrix},
\cdots ,
\color{myblue}{U_s}\color{black}
\begin{pmatrix}
 C_s \\G_s
\end{pmatrix}
\label{eqn2.18i}
\end{equation}
a Matrix Partition Code for source $\mathcal S$,
where
{\color{mbrown}$
 T = \begin{pmatrix}
      G_1 \\ \vdots \\ G_s 
     \end{pmatrix}
$, $C_i$ are row basis matrices of $Q_i$, and $U_i$ are arbitrary invertible matrices with appropriate sizes for all $i$}.
\label{defM}
\end{definition}

\color{myblue}
 Those $U_i$ may seem redundant and we usually set it to identity. But they are indispensable for the code to cover all perfect compression and linear-optimal compression.
\color{myblue}A Matrix Partition Code can be seen as a Pre-Matrix Partition Code with $T'=T$ and $G'_i=G_i$ for all $i$, and hence it is
a valid compression too.

\color{black}

This type of compression \eqref{eqn2.18i} first appeared in \cite{ma2011universality} to deal
with the multiple Hamming sources over $\mathbb Z_2$, in which we called it
Generalized HCMS for perfect compressions. Now we find that it is applicable to
any {\color{dblue} source with {\color{mygreen} deviation symmetry}}, a class of source much wider than Hamming source, over an 
arbitrary field. We would now call the code described by \eqref{eqn2.18i} as a {\color{dblue} Matrix Partition} Code and the matrix $P$ as the {\em parent matrix} of the {\color{dblue} Matrix Partition} Code.  We will show {\color{myblue} 
in Section \ref{sect:universality}} that
every compression of a source {\color{dblue} with {\color{mygreen} deviation symmetry}} can be deduced from Theorem \ref{thm2.1}.
Every linear-optimal or perfect compression is  {\color{dblue} a Matrix Partition Code}. {\color{dblue} Before doing that, we derive here the minimum  possible {\color{mbrown}sum-ratio} (highest compression) allowed by a Matrix Partition Code.}

\color{myblue}
\begin{thm} [{\color{mbrown}Compression Ratio} Tuples]
\label{thm2.2}
Suppose the parent matrix P of \eqref{eqn2.2} and \eqref{eqn2.3} is given. Then the 
 {\color{mbrown}compression ratio} tuples $(m_1/n,\cdots, m_s/n)$ of any Pre-Matrix Partition Code fulfill
$m_i=rank Q_i + r_i$,
where $r_i\in\{0,1,2,\cdot\cdot\}$ for all $i$ such that 
\begin{equation}
r_1+\cdots +r_s\ge n- 
rank\begin{pmatrix}
Q_1 \\
\vdots \\
Q_s
\end{pmatrix}.
\label{new1} 
\end{equation} 
Moreover equality \eqref{new1} holds if and only if the code is a Matrix 
Partition Code.
 \end{thm}
\color{purple}The proof below frequently uses the fact that
rank $A$+ rank$B\ge$ rank $\begin{pmatrix}A\\ B \end{pmatrix}$ and the equality
holds iff row$A\bigcap$ row$B=\{{\bf 0}\}$, where the word {\em row} means {\em the row space of}. \color{myblue}

\begin{proof}
Let $i\in\{1,\cdots,s\}$. Let $G'_i, Q_i, H_i$ be those defined in Theorem \ref{thm2.1}. 
WLOG, we decompose
\begin{equation}
G'_i=\begin{pmatrix}
A_i \\                D_i
\end{pmatrix},
\label{new1.5}
\end{equation}
such that
\begin{equation}
 \mbox{row }A_i\bigcap\mbox{row }Q_i=\{{\bf 0}\},\hspace{0.5cm} \mbox {row}\begin{pmatrix} A_i \\ Q_i \end{pmatrix}=\mbox{row }\begin{pmatrix} G'_i \\ Q_i \end{pmatrix}.
\label{new1.5.5}
\end{equation}
Equation \eqref{eqn2.7} becomes
 null$H_i=$null$\begin{pmatrix}A_i\\Q_i \end{pmatrix}$. Therefore
 row$H_i=$row$\begin{pmatrix}A_i\\Q_i \end{pmatrix}$ and 
 rank$H_i=$rank$\begin{pmatrix}A_i\\Q_i \end{pmatrix}$. Thus 
we have 
\begin{equation}
m_i=\mbox{number of rows of }H_i\ge\mbox{rank }H=
 \mbox{rank}A_i+\mbox{rank}Q_i. 
\label{new2}
\end{equation}
Moreover
$\begin{pmatrix}
Q_1 \\ \vdots \\   Q_s \\ A_1 \\ \vdots \\                 A_s
\end{pmatrix}$
 is injective (with $n$ columns) by \eqref{eqn2.5} and the second equation in \eqref{new1.5.5}. 
We get
\begin{equation}
\mbox{rank}A_1+\cdots+\mbox{rank}A_s\ge n-\mbox{rank}\begin{pmatrix}
Q_1\\ \vdots \\ Q_s \end{pmatrix}.
\label{new3}
\end{equation}
By \eqref{new2} and \eqref{new3}, we get \eqref{new1} with $r_i=m_i-$rank$Q_i$.

Secondly, equality \eqref{new1} holds iff equalities \eqref{new2} and \eqref{new3} both hold.
Let $G_i$ be a row basic matrix of $A_i$ and  $C_i$ be a row basic matrix of $Q_i$. 
We have {\color{mbrown}row $G_i\bigcap$row $C_i=\{{\bf 0}\}$}, thanks to the first equation in \eqref{new1.5.5}.
Equality \eqref{new2} holds iff $H_i$ is a surjective matrix and hence a row basic matrix of  $\begin{pmatrix}A_i\\Q_i \end{pmatrix}$.
Notice that  $\begin{pmatrix}G_i\\C_i \end{pmatrix}$ is also a surjective matrix of $\begin{pmatrix}A_i\\Q_i \end{pmatrix}$, we conclude that 
equality \eqref{new2} holds iff
\begin{equation} H_i=U_i\begin{pmatrix}G_i \\ C_i \end{pmatrix}
\label{new4}
\end{equation}
 for an invertible 
matrix $U_i$. We also
let $Y$ be a row basic matrix of  $\begin{pmatrix}Q_1\\ \vdots\\  Q_s \end{pmatrix}$. So 
$\begin{pmatrix}Y\\ G_1\\ \vdots\\  G_s \end{pmatrix}$ is injective and \eqref{new3} is equivalent to
rank$G_1+\cdots +$rank$G_s\ge n-$rank$Y$, which holds
iff 
row$G_1\oplus\cdots\oplus$row$G_s\oplus$row$Y=\mathbb F^n$. Since all $G_1,\cdots, G_s$ and $Y$ are 
surjective, we conclude equality \eqref{new3} holds iff
\begin{equation}
\begin{pmatrix} Y \\ T \end{pmatrix}\mbox{ is bijective, for }
T=\begin{pmatrix}G_1\\ \vdots \\ G_s\end{pmatrix}.
\label{new5}
\end{equation}
Notice that \eqref{new4} and \eqref{new5} {\color{mbrown}are all we need to define a Matrix Partition Code} with the given $P$ (c.f.
Definition \ref{defM}) and so we {\color{mbrown}complete the proof}.
\end{proof}

\begin{cor}
\label{cor1}
All the values of  {\color{mbrown}compression ratio} tuples allowed by \eqref{new1} are {\color{mbrown}achievable} by the 
Pre-Matrix Partition Codes. 
\end{cor}
\begin{proof}
Given any $r_1,\cdots, r_s$ such that equality \eqref{new1} holds. Let $T$ be a matrix defined in  \eqref{eqn2.18z} or  \eqref{new5}.
We have

\begin{equation}
r_1+\cdots + r_s = n-\mbox{rank}\begin{pmatrix} Q_1 \\ \vdots \\ Q_s \end{pmatrix}=\mbox{number of rows of } T.
\label{new6}
\end{equation}
Hence we can partition $T$ into those $G_i$ such that $G_i$ have $r_i$ rows for all $i$, respectively. Let
$(H_1,\cdots,H_s)$ be a compression of Matrix Partition Code defined by \eqref{eqn2.18i} or \eqref{new4}.
We have 
\begin{equation}
r_i=m_i-\mbox{rank}Q_i=\mbox{number of rows of } G_i.
\label{new7}
\end{equation}
Therefore $(H_1,\cdots,H_s)$ has the corresponding {\color{mbrown}compression ratio} tuples of the given $r_1,\cdots, r_s$.
As a result all the values allowed by the equality in \eqref{new1} are {\color{mbrown}achievable}.

In general, for any $(r_1,\cdots, r_s)$ allowed by \eqref{new1}, we let $a_i\in\{0,1,2,\cdots\}$ such that $r_i\ge a_i$ and
 \begin{equation}
a_1+\cdots +a_s= n- 
\mbox{rank}\begin{pmatrix}
Q_1 \\
\vdots \\
Q_s
\end{pmatrix}.
\label{new7.5} 
\end{equation}
By the previous argument, there exists a  Matrix Partition
Code $(H_1,\cdots, H_s)$ such that $m_i=a_i+$rank$Q_i$ for all $i$. Then 
$(H'_1,\cdots, H'_s)$ is the Pre-Matrix Partition Code with the {\color{mbrown}desired} {\color{mbrown}compression ratio} tuples,
where those $H'_i$ are obtained by augmenting the corresponding $H_i$ vertically with $r_i-a_i$ zero rows for all $i$. 
\end{proof}

\begin{cor}
\label{cor2} The total code length $M\ge$rank$Q_1+\cdots+$rank$Q_s-$rank$
\begin{pmatrix}Q_1 \\ \vdots \\ Q_s \end{pmatrix}$ for any Pre-Matrix Partition Code, and
the equality holds if and only if the code is a Matrix Partition Code. 
\end{cor}
\begin{proof}
Simply because $M=m_1+\cdots + m_s$.
\end{proof}

\color{black}
It is tempting to think that changing the choice of $\mathcal D$ should end up with a
different parent matrix $P$ that may increase compression {\color{nblue} efficiency}.
However, it turns out that it is not the case. The parent matrix $P$ is independent of such a choice as shown by the following theorem.

\begin{thm}
 \label{thm2.3}
Let both $\mathcal D$ and $\mathcal E$ be representation sets of source $\mathcal S$. If $P$ is a parent matrix of $\mathcal D$ as specified in Theorem \ref{thm2.1}, 
then $P|_{\tilde{\mathcal{E}}}$ is also injective, where $\tilde{\mathcal E}$ is a vectorized $\mathcal E$ given by

\begin{equation}
\tilde{\mathcal E}=\left\{  \left.\begin{pmatrix}
{\bf f}_1 \\
      \vdots \\
      {\bf f}_s      
     \end{pmatrix} \right| ({\bf f}_1,\cdots,{\bf f}_s)\in \mathcal E \right\}. 
\quad (c.f.  \, \eqref{eqn2.1})  
\label{eqn2.19}
\end{equation}
\end{thm}

\begin{proof}
From \eqref{eqn1.12}, there exists mapping ${\bf v}(\cdot)$ such that
\begin{equation}
\tilde{\mathcal E}=\left\{ \left.\begin{pmatrix}
{\bf v}(\delta)+{\bf d}_1 \\
       \vdots \\
      {\bf v}(\delta)+{\bf d}_s 
    \end{pmatrix}
 \right| \delta=({\bf d}_1,\cdots,{\bf d}_s)\in \mathcal D \right\}. 
\label{eqn2.20}
\end{equation}

Suppose 
\begin{equation}
P \begin{pmatrix}
 {\bf v}(\delta)+{\bf d}_1      \\
     \vdots \\
     {\bf v}(\delta)+{\bf d}_s 
  \end{pmatrix}
  = P
\begin{pmatrix}
{\bf v}(\gamma) + {\bf g}_1 \\
           \vdots \\
{\bf v}(\gamma) + {\bf g}_s                                             
\end{pmatrix},
\label{eqn2.21}
\end{equation}
where $\delta$, $\gamma\in D$ such that $\delta=({\bf d}_1,....,{\bf d}_s)$, $\gamma=({\bf g}_1,...,{\bf g}_s)$. Then
\[
P \begin{pmatrix}
 {\bf d}_1 \\ \vdots \\ {\bf d}_s  
  \end{pmatrix}
=  
  P 
\begin{pmatrix}
{\bf g}_1 \\ \vdots \\ {\bf g}_s 
\end{pmatrix}
\mbox{  by \eqref{eqn2.3} };
\]
\[
 \begin{pmatrix}
  {\bf d}_1 \\ \vdots \\{\bf d}_s
 \end{pmatrix}
= 
\begin{pmatrix}
 {\bf g}_1 \\ \vdots \\ {\bf g}_s
\end{pmatrix}
\mbox{ by {\color{black}\eqref{eqn2.2}}};
\]
   i.e.      $\delta=\gamma$
 and we get
\[
 \begin{pmatrix}
   {\bf v}(\delta)+{\bf d}_1        \\
     \vdots \\
     {\bf v}(\delta)+{\bf d}_s
 \end{pmatrix}
=
\begin{pmatrix}
    {\bf v}(\gamma) + {\bf g}_1 \\
     \vdots \\
         {\bf v}(\gamma) + {\bf g}_s
\end{pmatrix}.
\]
  Hence, $P|_{\tilde{\mathcal E}}$ is also injective.
\end{proof}

\section{\color{dblue} Perfect Compression of Matrix Partition Code}
\label{sect:perfect}

{\color{myblue} In this section, we study perfect compression. By Definition \ref{defP}, the field $\mathbb F$ is required to be finite, and a code is perfect if and only if the cardinality of the range of the mapping is the same as that of the source, i.e. $|\mathcal C|=|\mathcal S|$. Since $|\mathcal C|=|\mathbb F|^M$ {\color{nbrown}and} $|\mathcal S|=|\mathbb F|^n|\mathcal D|$ in 
\eqref{eqn1.4} and $|\mathcal D|=|\tilde {\mathcal D}|$ in \eqref{eqn1.15}, we have
\begin{equation}
|\mathbb F|^{M-n}=|\mathcal D|=|\tilde {\mathcal D}|.
\label{eqn6.0}
\end{equation}
For the sake of simplicity, we won't extend the definition to infinite 
field.
\color{black}

\subsection{Perfect Compression for $\mathcal S$ over Finite Fields}

The following theorem explains a necessary condition for a perfect {\color{dblue} Matrix Partition Code}.

\color{myblue}
\begin{thm}[Necessary Condition of Perfect Codes]
\label{thm6}
To have a perfect  compression of Matrix Partition Code, we must have an $(M-n)\times sn$ matrix $P$ 
such that
\begin{equation}
P|_{\color{myblue}\tilde{\mathcal D}\color{black}} \mbox{ is bijective.}
\label{eqn6.1}  
\end{equation}
\end{thm}

\begin{proof}
Suppose we have a perfect code constructed from a parent matrix $P'$, any matrix P with the same row space of $P'$ can be viewed as the parent matrix of the code. 
\color{myblue}Indeed if row$P=$row$P'$, then null$P=$null$P'$ and row$Q_i=$row$Q'_i$ for all $i$.
Hence $P$ also satisfies \eqref{eqn2.2} and \eqref{eqn2.3}, and shares the same other components (such as $Y$, $C_i$,$\cdots$, etc.)
with $P'$ in Definition \ref{defM}.
\color{myblue}
In particular, we
let $P$ be a row basic matrix of $P'$. Let $r$ be the number of rows of $P$ so that $P$ is an $r\times sn$ matrix. We are going to show $P|_{\tilde{\mathcal D}}$ is bijective and $r=M-n$.
 
{\color{nblue} Since we have shown $P|_{\tilde{\mathcal D}}$ is injective (that is \eqref{eqn2.2}), we only need to show that $P|_{\tilde{\mathcal D}}$ is also surjective. }
Suppose $P|_{\tilde{\mathcal D}}$ is not surjective, then we can pick a ${\bf u}\in\mathbb F^r$ such that ${\bf u}\notin P(\tilde{\mathcal D})$. Since $P$ {\color{nblue} is a row basis matrix and thus} is a surjective matrix, there exists a ${\bm \delta}\in\mathbb F^{sn}$ with $P{\bm \delta}={\bf u}$. Notice that {\color{nblue} by \eqref{eqn2.3}},
\begin{equation}
P\left({\bm \delta} +\begin{pmatrix}\bf v\\ \vdots \\ \bf v\end{pmatrix}\right)=P({\bm \delta})={\bf u}\mbox{ for all }\bf v\in\mathbb F^n,
\label{eqn6.2}
\end{equation}
thus ${\bm \delta} +\begin{pmatrix}\bf v\\ \vdots \\ \bf v\end{pmatrix}\notin \tilde{\mathcal D}$ for all ${\bf v}\in\mathbb F^n$. 
Therefore we can extend $\tilde{\mathcal D}$ to $\tilde{\mathcal D}'=\tilde{\mathcal D}\cup\{{\bm \delta}\}$ and the source $\mathcal S$ to the corresponding $\mathcal S'$. Notice that 
$P|_{\tilde{\mathcal D}'}$ is injective and hence we can compress $\mathcal S'$ by the same compression. This leads to a contradiction as $|\mathcal S'|>|\mathcal S|=|\mathcal C|$. 

Finally, by \eqref{eqn6.0}, we must have $r=M-n$ if $P|_{\tilde{\mathcal D}}$ is bijective.
\end{proof}

Actually \eqref{eqn6.1} is a necessary condition for any perfect compression of the given $\mathcal S$ simply because it turns out that any perfect compression can be realized by a Matrix Partition Code. We will defer the discussion to Section \ref{sect:universality}.

\color{black}

\subsection{Necessary Conditions for Perfect Compression on Generalized Hamming Source}
Let $\mathcal L$ be a non-empty subset of $\mathbb F$ s.t. $0\notin \mathcal L$. We define
\begin{equation}
\mathcal S=\{({\bf v},\cdots,{\bf v})+(\underbrace{{\bf 0},\cdots,\lambda {\bf e}_j}_{i-th},\cdots,{\bf 0}) | {\bf v}\in\mathbb F^n, \lambda\in \mathcal L \cup \{0\}, 1\le i \le s, 1\le j\le n\}.
\label{eqn7.1}                                                         
\end{equation}

Notice that if $\mathcal L=\mathbb F-\{0\}$, then $\mathcal S$ is just the Hamming source over $\mathbb F$ (c.f. \eqref{eqn1.5}).
Therefore we call $\mathcal S$ as generalized Hamming source. Obviously it is a {\color{dblue} source with {\color{mygreen} deviation symmetry}}.

Let $s \ge 3$. We pick
\begin{equation}
\mathcal D=\{(\underbrace{{\bf 0},\cdots ,\lambda {\bf e}_j}_{i-th},\cdots, {\bf 0}) | \lambda\in \mathcal L\cup 
\{0\},  1 \le i \le s, 1 \le j \le n\}
\label{eqn7.3} 
\end{equation}
and the corresponding
\begin{equation}
\color{myblue}\tilde{\mathcal D}\color{black}=\{\lambda {\bf e}_i | \lambda\in \mathcal L\cup \{0\}, 1 \le i \le sn\}.                                              
\label{eqn7.4}
\end{equation}
We have
\begin{equation}
|\tilde{\mathcal D}|=1+|\mathcal L|sn.                                                                 
\label{eqn7.5}
\end{equation}
To have a perfect compression, we must have \eqref{eqn6.0} and hence
\begin{equation}
|\mathbb F|^{M-n}=1+|\mathcal L|sn.                                                          
\label{eqn7.6}
\end{equation}
 So, $s$ and $|\mathcal L|$ can't be multiplier of $p$, the characteristic 
of $\mathbb F$ ($|\mathbb F|=p^u$ for some positive integer $u$). If it is the case,
 then we have infinite pair of
numbers $(M,n)$ satisfying \eqref{eqn7.6} by Euler theorem.

\begin{thm}[{\color{dblue} Necessary} Conditions of Perfect {\color{dblue} Matrix Partition Codes} for Generalized Hamming Sources]
\label{thm7}
The necessary and sufficient condition for the existence of an 
$(M-n)\times sn$ matrix $P$ which is bijective when restricted to $\color{myblue}\tilde{\mathcal D}\color{black}$ is that 
$|\mathbb F|-1$ is divisible by $|\mathcal L|$ 
and $\exists$ distinct $a_1, a_2,\cdots ,a_k \in \mathbb F$, with $k=(|\mathbb F|-1)/|\mathcal L|$, such
that 
\begin{equation}
\mathbb F-\{0\}=\{a_i \lambda | 1 \le i \le k; \lambda\in \mathcal L\}.         
\label{eqn7.9}
\end{equation}
\end{thm}

\begin{proof}
If $\mathcal L$ fulfills the conditions, 
then $sn/k=(|\mathbb F|^{M-n}-1)/(|\mathbb F|-1)$ by \eqref{eqn7.6}. Thus $sn/k$ is an integer. Let
$\{{\bf v}_1, {\bf v}_2, \cdots ,{\bf v}_{sn/k}\}$  be a subset of $\mathbb F^{M-n}$ that each element is a 
not multiplier of the 
other.
Define $P$ through its column
\begin{equation}
{\bf P}_{i+(j-1)k}=a_i {\bf v}_j,\mbox{ for } 1 \le i \le k, 1 \le j \le sn/k.                                 
\label{eqn7.10}
\end{equation}
Then we will show $P|_{\color{myblue}\tilde{\mathcal D}\color{black}}$ is injective. Suppose
\begin{equation}
P(\lambda_1 {\bf e}_{i+(j-1)k})=P(\lambda_2 {\bf e}_{b+(c-1)k}),                           
\label{eqn7.11}
\end{equation}
where $\lambda_1, \lambda_2\in \mathcal L\cup\{0\}; i,b \le  k; j,c \le sn/k$.
Then 
\begin{equation}
\lambda_1a_i{\bf v}_j=\lambda_2a_b{\bf v}_c,                                             
\label{eqn7.12}
\end{equation}
which gives $\lambda_1=\lambda_2=0$ that yields $\lambda_1 {\bf e}_{i+(j-1)k}=\lambda_2 {\bf e}_{b+(c-1)k}$ immediately or $j=c$ with $\lambda_1\neq0\neq\lambda_2$. So let assume we are in the second case. By counting the number of elements in both sides of 
\eqref{eqn7.9}, 
we conclude that every nonzero element of $\mathbb F$ is a product of 
a unique $\lambda$ and a unique $a_j (\lambda \in \mathcal L, 1 \le j \le k)$. Thus, 
we get $\lambda_1=\lambda_2$ and $i=b$. Hence $P|_{\color{myblue}\tilde{\mathcal D}\color{black}}$ is injective. By 
\eqref{eqn6.0}, $P|_{\color{myblue}\tilde{\mathcal D}\color{black}}$ is bijective.

Conversely, let $\mathcal J$ be a maximal subset of the index set 
$\{ 1,2, \cdots, sn \} $ such that ${\bf P}_j$, the j-th column of $P$, is a multiplier of ${\bf P}_1$ for
 all $j\in\mathcal J$. So if ${\bf P}_i$ is a multiplier of ${\bf P}_1$, then $i\in\mathcal J$. 
Let ${\bf P}_j=a_j{\bf P}_1$ for all $j\in\mathcal J$. The columns of $P$ must be nonzero 
and distinct from each other, otherwise $P|_{\color{myblue}\tilde{\mathcal D}\color{black}}$ can't be injective. It follows 
that $a_j$ are nonzero for all $j\in\mathcal J$ and distinct from each other. 
Then 
the bijectivity of 
$P|_{\color{myblue}\tilde{\mathcal D}\color{black}}$ implies $\forall b\in\mathbb F-\{0\}, \exists$ unique $\lambda 
{\bf e}_i\in 
\color{myblue}\tilde{\mathcal D}\color{black}$ (c.f. \eqref{eqn7.4}) with $\lambda\in\mathcal L, 1\le i\le sn$, such that $P( \lambda  {\bf e}_i)=b{\bf P}_1$. 
Hence $i\in\mathcal J$ and $\lambda a_i=b$. By counting, we get $(|\mathbb F|-1)/|\mathcal 
L|$ is an integer and \eqref{eqn7.9} is fulfilled with $k=|\mathcal J|$.
 \end{proof}
We remark that {\color{nblue} Theorem \ref{thm7} only characterizes necessary conditions since even if $P$ is bijective when restricted to $\tilde{\mathcal D}$,  it does not mean that we will have a perfect compression of Matrix Partition Codes (see \cite{ChengM:DCC2010}). 
However, it is not the case when $s=2$.
Here we give some examples of perfect compression:

\begin{example}
\label{example1}
$\mathbb F=\mathbb Z_{11}, \mathcal L=\mathbb Z_{11}-\{0\}$ (Hamming source over $\mathbb Z_{11}$), $n=4$ and $s=3$:
\begin{align}
H_1 & =Q_1=\begin{pmatrix}
1 & 1 & -2 &  -2 \\
0 & 2 & -1 & -7
               \end{pmatrix}, \\
H_2 & =Q_2=
\begin{pmatrix}
0 & 9 & 1 & 1 \\
1 & 5 & 5 & 8
\end{pmatrix}, \\
H_3 & =Q_3=
\begin{pmatrix}
-1 & 1 & 1 & 1 \\
-1 & 4 & 7 & -1
\end{pmatrix}.
\end{align}
Notice that each nonzero vector of $\mathbb F^2$ has one and only one multiplier 
as a column vector of 
$
P= 
[\mbox{
\begin{tabular}{ @{} c @{} ;{2pt/2pt} @{} c @{} ;{2pt/2pt} @{} c @{}}
 $\hspace{-0.1cm} Q_1$ & $\hspace{0.05cm} Q_2 \hspace{0.02cm}$ & $\hspace{0.05cm} Q_3 \hspace{-0.1cm}$
\end{tabular}
}].
$
\end{example}

\begin{example}
$\mathbb F=\mathbb Z_5, \mathcal L=\{1,-1\}, n=4$ and $s=3$:
\begin{align}
H_1 & =Q_1= \begin{pmatrix}
1 & 0 & 1 & 1 \\
0 & 2 & 2 & -2
                \end{pmatrix}, \\
H_2 & =Q_2=
\begin{pmatrix}
0 & 2 & -2 & 2 \\
1 & 0 & -1 & -1
\end{pmatrix}, \\
H_3 & =Q_3=
\begin{pmatrix}
-1 & -2 & 1 & 2 \\   
-1 & -2 & -1 & -2
\end{pmatrix},
\end{align}
Notice that $\{a_1,a_2\}=\{1,2\}$ (c.f.\eqref{eqn7.9}),\newline $\{{\bf v}_1,{\bf v}_2,{\bf v}_3,{\bf v}_4,{\bf v}_5,{\bf v}_6\}= 
\left\{ 
\begin{pmatrix} 1 \\ 0 \end{pmatrix},
\begin{pmatrix} 0 \\ 1 \end{pmatrix},
\begin{pmatrix} -1 \\ -1\end{pmatrix},
\begin{pmatrix} 1 \\ 2 \end{pmatrix},
\begin{pmatrix} -2 \\ -1 \end{pmatrix},
\begin{pmatrix} 1 \\ -1 \end{pmatrix}
\right\}$ (c.f.\eqref{eqn7.10}).
\end{example}

\begin{example}
\label{example3}
$\mathbb F=\mathbb Z_5, \mathcal L=\{1\}, n=6, s=4$:
\begin{align}
H_1 &=Q_1=
\begin{pmatrix}
1 & 0 & 1 & 1 &  2 &  2 \\
0 & 1 & 1 & 2 & -2 & -1
\end{pmatrix}, \\
H_2 &=Q_2=
\begin{pmatrix}
-1 & 0 & 2 & -1 & 2 & -1 \\
 0 & 2 & 2 & -2 & 1 & 1
\end{pmatrix}, \\
H_3 &=Q_3=
\begin{pmatrix}
2 & 0 & -2 &  1 & -2 & 1 \\
0 & -1 & -2& -2 & 2 & -1
\end{pmatrix}, \\
H_4 &=Q_4=
\begin{pmatrix}
-2 & 0 & -1 & -1 & -2 & -2 \\
 0 &-2 & -1 & 2 & -1 & 1
\end{pmatrix}.
\end{align}
The matrix 
$P=[\mbox{
\begin{tabular}{ @{} c @{} ;{2pt/2pt} @{} c @{} ;{2pt/2pt} @{} c @{} ;{2pt/2pt} @{} c @{}}
 $\hspace{-0.1cm} Q_1$ & $\hspace{0.05cm} Q_2 \hspace{0.02cm}$ & $\hspace{0.05cm} Q_3 \hspace{0.02cm}$ & $\hspace{0.05cm} Q_4 \hspace{-0.1cm}$
\end{tabular}
}]$
consists of all nonzero vectors of $\mathbb F^2$ 
without repetition.
\end{example}

\begin{example}
$\mathbb F = \mbox{GF}(4)=\mathbb Z_2(\alpha)$ with $\alpha^2+\alpha+1=0$;
 $\mathcal L = \{1, \alpha, \alpha+1 \}$, $n=7, s=3$ 
\begin{align}
 T & = \begin{pmatrix}
        0 & 0 & 0 & 0 & 0 & 0 & 1
       \end{pmatrix},
 \\
 Q_1  & =
 \begin{pmatrix}
 1 & 1 & 0 & 0 & 1 & \alpha & \alpha \\
  1 & 0 & 1 & 0 & \alpha & 1 & \alpha \\
  0 & 0 & 0 & 1 & \alpha & \alpha & 1
 \end{pmatrix}, \\
 H_1 & = \begin{pmatrix}
          T \\ Q_1
         \end{pmatrix},
\\
 H_2 & = Q_2  = 
 \begin{pmatrix}
  1 & \alpha+1  & 0 & 1 & \alpha & \alpha+1  & \alpha+1  \\
  0 &  1 & \alpha+1  &0 & 1 & \alpha+1  & 1 \\
  1 & 0 &  1& \alpha+1  & 1 & 1 & \alpha+1 
 \end{pmatrix},
\\
 H_3 & = Q_3  = 
 \begin{pmatrix}
 0 & \alpha & 0 & 1 & \alpha+1  & 1 & 1\\
 1 & 1 & \alpha & 0 & \alpha+1  & \alpha & \alpha+1  \\
 1 & 0 & 1 & \alpha & \alpha+1  & \alpha+1  & \alpha
\end{pmatrix}.
\end{align}
It is a Hamming source over 
GF$(4)$. Each nonzero vector of $\mathbb F^3$ has one and only one multiplier 
as a column vector of 
$
P= 
[\mbox{
\begin{tabular}{ @{} c @{} ;{2pt/2pt} @{} c @{} ;{2pt/2pt} @{} c @{}}
 $\hspace{-0.1cm} Q_1$ & $\hspace{0.05cm} Q_2 \hspace{0.02cm}$ & $\hspace{0.05cm} Q_3 \hspace{-0.1cm}$
\end{tabular}
}].
$
Besides, $(Q_1, 
\begin{pmatrix}
          T \\ Q_2
         \end{pmatrix}, Q_3)$ and $(Q_1, Q_2, 
\begin{pmatrix}
          T \\ Q_3
         \end{pmatrix})$ are also perfect compressions.
\end{example}

Let $s=2$. We rewrite \eqref{eqn7.1} as
\begin{equation}
\mathcal S=\{({\bf v},{\bf v})+({\bf 0}, a {\bf e}_i) | {\bf v}\in \mathbb F^n, 1 \le i\le n, a\in \mathcal L\cup\mathcal L_{-}\cup\{0\}\},       
\label{eqn7.13}
\end{equation}
where $\mathcal L_{-}=\{-a | a\in \mathcal L\}$. We have
\begin{equation}
\mathcal D=\{({\bf 0}, a {\bf e}_i)) | 1 \le i \le n, a\in \mathcal L\cup\mathcal L_{-}\cup\{0\}\}.                            
\label{eqn7.14}
\end{equation}
The corresponding
\begin{equation}
\color{myblue}\tilde{\mathcal D}\color{black}=\{a {\bf e}_i | a\in \mathcal L\cup\mathcal L_{-}\cup\{0\}, n<i \le 2n\}.
\label{eqn7.15}
\end{equation}
Then 
\begin{equation}
|\tilde{\mathcal D}|=1+n|\mathcal L\cup\mathcal L_{-}|.                                                          
\label{eqn7.16}
\end{equation}
To have a perfect compression, we must have \eqref{eqn6.0}, i.e.
\begin{equation}
|\mathbb F|^{M-n}=1+n|\mathcal L\cup\mathcal L_{-}|.                                                        
\label{eqn7.17}
\end{equation}
Now we are seeking an $(M-n) \times 2n$ matrix $P$ to be bijective when restricted to $\color{myblue}\tilde{\mathcal D}\color{black}$.
Since the first $n$ columns in $P$ virtually play no role on $\color{myblue}\tilde{\mathcal D}\color{black}$, it can be arbitrary. Let
\begin{equation}
P= 
[\mbox{
\begin{tabular}{ @{} c @{} ;{2pt/2pt} @{} c @{} }
 $\hspace{-0.1cm} -Q_2 \hspace{0.03cm}$ &  $\hspace{0.08cm} Q_2 \hspace{-0.1cm}$
\end{tabular}
}].
\label{eqn7.20}
\end{equation}
where $Q_2$ is an $(M-n) \times n$ matrix so that $P$ satisfies \eqref{eqn2.3}. Let
\begin{equation}
\color{myblue}\tilde{\mathcal D}\color{black}'=\{a {\bf e}_i | a\in \mathcal L\cup\mathcal L_{-}\cup\{0\}, 1 \le i \le n \},                                   
\label{eqn7.21}
\end{equation}
which is just the nontrivial segment of the $\color{myblue}\tilde{\mathcal D}\color{black}$ in \eqref{eqn7.15}.

\begin{thm}[{\color{dblue} Necessary and Sufficient Conditions of Perfect {\color{dblue} Matrix Partition Codes} for Generalized Hamming Sources with $s=2$}]
\label{thm7.2}
The following statements imply each other:
\begin{itemize}
 \item 
We have  an $(M-n) \times 2n$ matrix $P$ which is bijective when restricted to $\color{myblue}\tilde{\mathcal D}\color{black}$.
\item
 We have  an $(M-n) \times n$ matrix $Q_2$ which is bijective when restricted to
$\color{myblue}\tilde{\mathcal D}\color{black}'$.
\item
 $\exists$ distinct $a_1, a_2,\cdots ,a_k \in \mathbb F$, with $k=(|\mathbb F|-1)/|\mathcal L\cup \mathcal L_{-}|$,
such that 
$\mathbb F-\{0\}=\{a_i\lambda | 1 \le i \le k; \lambda\in \mathcal L\cup\mathcal L_{-}\}$.                 
\end{itemize}
\end{thm}

\begin{proof}
 Similar to the {\color{myblue} proof of Theorem \ref{thm7} }. 
\end{proof}

Once we have the $Q_2$ in the above theorem, it can be shown that the pair
$\begin{pmatrix}
 G_1 \\
 Q_2
\end{pmatrix}$,
$\begin{pmatrix}
 G_2 \\
 Q_2
\end{pmatrix}$ is a {\color{dblue} Matrix Partition Code} and a perfect compression whenever
$\begin{pmatrix}
 G_1 \\
 G_2 \\
  Q_2
\end{pmatrix}$ forms an $n\times n$ invertible matrix. Notice that
$(I_{n\times n},Q_2)$ is also a perfect compression, which is 
a {\color{dblue} Matrix Partition Code} with certain $U_1$.


\begin{example}
$\mathbb F=GF(4)$ with $\alpha^2+\alpha+1=0$, $\mathcal L=\mathbb F-\{0\}$, $n=5$, $s=2$.
Let 
{\color{dblue}
\begin{align}
Q_2=\begin{pmatrix}
  1 & 0 & 1 & 1 &  1 \\
                0 &1 &1 &\alpha &1+\alpha         
         \end{pmatrix}.
\end{align}}
The following pairs of matrices are all perfect compression:
\begin{itemize}
 \item 
$I_{5 \times 5}$ and $Q_2$;
\item
$
\begin{pmatrix}
\begin{smallmatrix}
0  & 0 & 1 & 0 & 0 \\
 0  & 0 &  0 & 1 & 0 \\
 0 & 0 & 0 & 0 & 1 
\end{smallmatrix} \vspace{0.05cm} \\ 
\hdashline[1pt/5pt]
 Q_2
\end{pmatrix}
$
and $Q_2$;
\item
$
\begin{pmatrix}
\begin{smallmatrix}
0  & 0 & 1 & 0 & 0 \\
 0 & 0 & 0 & 0 & 1 
\end{smallmatrix} \vspace{0.05cm} \\ 
\hdashline[1pt/5pt]
 Q_2
\end{pmatrix}
$
and 
$
\begin{pmatrix}
\begin{smallmatrix}
 0 & 0 & 0 & 1 & 0 
\end{smallmatrix} \vspace{-0.05cm} \\ 
\hdashline[1pt/5pt]
 Q_2
\end{pmatrix}
$.

\end{itemize}

%
 
\end{example}

\color{myblue}
\subsection{Examples beyond Generalized Hamming Source}
Here we will provide some examples where the sources are not generalized Hamming. The first two examples illustrate that one can modify a given compression when the original source has been deformed.

\begin{example}
$\mathbb F=\mathbb Z_{11}$, $s=5$, $n=6.$

$\mathcal D=\{(\underbrace{{\bf 0},\cdots,{\bf 0},a{\bf e}_j}_{i-th},{\bf 0},\cdots,{\bf 0})\hspace{0.1cm} | \hspace{0.1cm}a \in\mathbb F, 1\le i\le3, 1\le j \le 4\}.$

Notice that this is just the Hamming source (c.f. Example \ref{example1}) trapping in a bigger space. So we can modify the previous setting to obtain a new perfect compression. 
\begin{align}
H_1 & =Q_1=\begin{pmatrix}
1 & 1 & -2 &  -2  & 0 & 0\\
0 & 2 & -1 & -7 & 0 & 0
               \end{pmatrix}, \\
H_2 & =Q_2=
\begin{pmatrix}
0 & 9 & 1 & 1 & 0 & 0\\
1 & 5 & 5 & 8 & 0 & 0
\end{pmatrix}, \\
H_3 & =Q_3=
\begin{pmatrix}
-1 & 1 & 1 & 1&0 & 0 \\
-1 & 4 & 7 & -1 & 0 & 0
\end{pmatrix},\\
Q_4&=Q_5=
\begin{pmatrix}
0 & 0 & 0 & 0&0 & 0 \\
0 & 0 & 0 & 0 & 0 & 0
\end{pmatrix},\\
\hspace{0.5cm}& T=
\begin{pmatrix}
0 & 0 & 0 & 0&1 & 0 \\
0 & 0 & 0 & 0 & 0 & 1
\end{pmatrix}.
\end{align}
And we chose $H_4=\begin{pmatrix}0 & 0 & 0 & 0 & 1 & 0\end{pmatrix}$, 
$H_5=\begin{pmatrix}0 & 0 & 0 & 0 & 0 & 1\end{pmatrix}$.
\end{example}

\begin{example}
$\mathbb F=\mathbb Z_{11}$, $s=4$, $n=5.$

\begin{align}
\mathcal D =& \{(\underbrace{{\bf 0},\cdots,{\bf 0},a{\bf e}_j}_{i-th},{\bf 0},\cdots,{\bf 0})\hspace{0.1cm} | \hspace{0.1cm}a \in\mathbb F, 1\le i\le 3, 1\le j \le 5\}\\
   & -\{(a{\bf e}_1,{\bf 0},{\bf 0},{\bf 0}),\hspace{0.1cm}({\bf 0},a{\bf e}_1,{\bf 0},{\bf 0}),\hspace{0.1cm}({\bf 0},{\bf 0},a{\bf e}_2,{\bf 0})\hspace{0.1cm}|\hspace{0.1cm}a\in\mathbb F\}.
\end{align}
 This is the source of the Example \ref{example1} with some shifting. We extend $P=$ \begin{tabular}{ @{} c @{} ;{2pt/2pt} @{} c @{} ;{2pt/2pt} @{} c @{} ;{2pt/2pt} @{} c @{}}
 $\hspace{-0.1cm} [Q_1$ & $\hspace{0.05cm} Q_2 \hspace{0.02cm}$ & $\hspace{0.05cm} Q_3 \hspace{0.02cm}$& $\hspace{0.05cm} Q_4 ]$
\end{tabular} accordingly.
\begin{align}
H_1 & =Q_1=\begin{pmatrix}
1 & 1 & -2 &  -2 &1 \\
0 & 2 & -1 & -7 & 0
               \end{pmatrix}, \\
H_2 & =Q_2=
\begin{pmatrix}
0 & 9 & 1 & 1 & 0\\
1 & 5 & 5 & 8 & 1
\end{pmatrix}, \\
H_3 & =Q_3=
\begin{pmatrix}
-1 & 1 & 1 & 1 & 1\\
-1 & 4 & 7 & -1 & 4
\end{pmatrix},
\end{align}
$Q_4=-Q_1-Q_2-Q_3=
\begin{pmatrix}
0& 0 & 0 & 0 & -2\\
0 & 0 & 0 & 0 & -5
\end{pmatrix}$, and $H_4=\begin{pmatrix}0&0&0&0&1\end{pmatrix}$, a row basis matrix of $Q_4$.

\end{example}

In the third example, we make use of an existing code to create a compression for another source, where $\mathcal{D}$ has been changed almost completely. The old code works as long as the parent matrix $P$ still fulfills \eqref{eqn2.2} with the new $\tilde{\mathcal D}$. 
If the existing one is a perfect compression and $P|_{\mbox{new }\tilde{\mathcal D}}$ is bijective, then the
compression is also perfect  for the new source simply by counting.

\begin{example}
$\mathbb F=\mathbb Z_{5}$, $n=6$, $s=4.$
\begin{align}
\mathcal D = & \{(\pm {\bf e}_j, {\bf 0}, {\bf 0}) | 1\le j \le 6\}\cup\{({\bf 0}, {\bf e}_j, {\bf 0}, {\bf 0})|  j \in\{2,3,5, 6\}\}\cup\nonumber\\ 
                    &\{({\bf e}_3,{\bf 0},{\bf e}_1,{\bf e}_3), ({\bf e}_1+{\bf e}_2, {\bf 0}, {\bf e}_3, {\bf e}_3),  ({\bf 0}, {\bf e}_1+{\bf e}_2+{\bf e}_4,{\bf 0},{\bf 0})\}\cup\nonumber\\
                    &\{(3{\bf e}_1, 2{\bf e}_2,-2{\bf e}_1,{\bf e}_4+{\bf e}_2), ({\bf 0},{\bf 0},{\bf 0},{\bf e}_4),({\bf e}_5,{\bf e}_6,{\bf 0},{\bf 0})\}\cup\nonumber\\
                    &\{({\bf 0},{\bf 0},{\bf e}_4,{\bf 0}),({\bf 0},{\bf 0},{\bf 0},{\bf e}_2), ({\bf 0},{\bf 0},{\bf 0},{\bf 0})\}.
\end{align}
The compression $(H_1,H_2,H_3,H_4)$ is the same as Example \ref{example3}.

\end{example}

\color{black}

\section{Uniqueness of {\color{dblue} Matrix Partition Codes}}
\label{sect:universality}

{\color{myblue} In this section, we will show that {\color{dblue} Matrix Partition Codes} are unique in the sense that any linear-optimal or perfect compression is a {\color{dblue} Matrix Partition Code}. 

\subsection{Null Space View}

We will first study the null spaces of lossless compression simply
because null spaces of coding matrices determines injectivity entirely.

\begin{lem}

\label{lem8.1}
If $(H_1,\cdots,H_s)$ is a lossless compression of a {\color{dblue} source with {\color{mygreen} deviation symmetry}}, then we have
$\mbox{null} H_1\cap\cdots\cap \mbox{null} H_s=\{{\bf 0}\}$.
 
\end{lem}

\begin{proof}
Let ${\bf v}\in \mbox{null} H_1\cap\cdots \cap \mbox{null} H_s$. Pick a $\sigma \in \mathcal S$. We have 
$\sigma + ({\bf v},\cdots,{\bf v}) \in \mathcal S$ (c.f. 
{\color{nblue} Definition \ref{def:szeto}}).
Moreover,
$(H_1,\cdots,H_s)(\sigma) = (H_1,\cdots,H_s)(\sigma + ({\bf v},\cdots,{\bf v}))$.
As $(H_1,\cdots,H_s)$ is a compression, $\sigma=\sigma+({\bf v},\cdots,{\bf v})$ that ${\bf v}={\bf 0}$.
 \end{proof}

\color{myblue}
Thus if $\mbox{null} H_1\cap\cdots \cap\mbox{null} H_{s-1}\supset K\ne\{{\bf 0}\}$, we have null$H_s\cap K=\{{\bf 0}\}$. In this situation, the following theorem tells us that we can build up another compression $H'_1,\cdots, H'_s$ merely by shifting the $K$ from the one of the first $s-1$ terminals to the last terminal.
\color{black}

\begin{thm}[Nullspace Shifting]
\label{thm8.1}
Suppose $(H_1,\cdots,H_s)$ is a compression for $\mathcal S$. Let $\pi$ be a permutation of the 
index set $\{1,2,\cdots,s\}$ that
\begin{align}
\left\{
\begin{matrix}
\mbox{null} H_{\pi(i)}& = & K\oplus N_i, & \mbox{ for $1\le i<s$,} \\               
\mbox{null} H_{\pi(s)}& = & N_s,                                 &                    
 \end{matrix}
\right.
 \label{eqn8.1}
\end{align}
where $K, N_i$ are subspaces of $\mathbb F^n$. 
Then $(H'_1,\cdots,H'_s)$ is also a compression for $\mathcal S$ if 
\begin{eqnarray}
\left\{
\begin{matrix}
\mbox{null} H'_{\pi(1)}& = &N_1,   & \\
 \mbox{null} H'_{\pi(i)}& =& K\oplus N_i, & \mbox{ for $1<i \le s$}.                    
 \end{matrix}
\right.
\label{eqn8.2}
\end{eqnarray}

Furthermore if $H_i$ and $H'_i$ are surjective for all $i$, then the two compressions 
have the {\color{dblue} same compression {\color{mbrown}sum-ratio}}. For finite field that if $(H_1,\cdots,H_s)$ is a perfect compression, then
$(H'_1,\cdots,H'_s)$ is also a perfect compression.
\end{thm}

\begin{proof}[Proof of Theorem \ref{thm8.1}]
WLOG, let's put $\pi=1$ for simplicity. {\color{nblue} Note that $(H'_1,\\\cdots,H'_s)|\mathcal{S}$ is injective if and only if $\mbox{null} H'_1 \times\cdots\times \mbox{null} H'_s \cap \mathcal S_+ = \{{\bf 0} \}$}, where {\color{dblue}
$\mathcal S_+=\{\sigma_1 - \sigma_2 | \sigma_1, \sigma_2\in \mathcal S\}$.\footnote{ 
\color{dblue} {Here we use the notation $S_+$ for the sake of consistence with 
\cite{ma2011universality}, where $\{\sigma_1 - \sigma_2 | \sigma_1, \sigma_2\in \mathcal S\} = 
\{\sigma_1 + \sigma_2 | \sigma_1, \sigma_2\in S\}$ since only binary sources were considered.}
}} Let
$\sigma_+\in \mbox{null} H'_1 \times\cdots\times \mbox{null} H'_s \cap \mathcal S_+$. Decompose
$\sigma_+=({\bf n}_1, {\bf k}_2+{\bf n}_2, \cdots, {\bf k}_{s-1}+{\bf n}_{s-1}, {\bf k}_s+{\bf n}_s)$, where
${\bf n}_i\in N_i$ and ${\bf k}_i\in K$ for all $i$. Since $\mathcal{S}_+$ is also a {\color{dblue} source with {\color{mygreen} deviation symmetry}}, we have
\begin{equation}
\sigma_+ - ( {\bf k}_s,\cdots,{\bf k}_s)=
   ({\bf n}_1-{\bf k}_s, {\bf k}_2+{\bf n}_2-{\bf k}_s,\cdots,{\bf k}_{s-1}+{\bf n}_{s-1}-{\bf k}_s, {\bf n}_s) \in \mathcal S_+.  
\label{eqn8.3}
\end{equation}
By checking the null space of $(H_1,\cdots,H_s)$, we find
$\sigma_+ - ( {\bf k}_s,\cdots,{\bf k}_s)\in \mathcal S_+ \cap \mbox{null} H_1 \times \cdots \times \mbox{null} H_s=\{{\bf 0}\}$
because $(H_1,\cdots,H_s)|_\mathcal S$ is injective. The first entry in the RHS of \eqref{eqn8.3}
gives ${\bf n}_1={\bf k}_s={\bf 0}$, and all other entries follow suit and give ${\bf k}_i={\bf 0}={\bf n}_i$ for all $i$.
Hence $\sigma_+={\bf 0}$ and $(H'_1,\cdots,H'_s)|_\mathcal S$ is injective.

Next if all $H_i$ and $H'_i$ are surjective for all $i$, then the two compressions have the same {\color{dblue} compression {\color{mbrown}sum-ratio}
$(sn-\sum_{i=1}^s dim N_i-(s-1)dim K)/n$.}

Lastly, if $(H_1,\cdots,H_s)$ is a perfect compression, then $(H_1,\cdots,H_s)|_\mathcal S$  is
surjective, still more {\color{black} so for} $(H_1,\cdots,H_s)$ per se. If $(H'_1,\cdots,H'_s)$ is
also surjective, the codeword spaces of the two compressions will be of the same
dimension $sn-\sum_{i=1}^s dim N_i-(s-1)dim K$ and hence same cardinality.
\end{proof}


\subsection{Proof of Uniqueness of {\color{dblue} Matrix Partition Codes}}

{\color{myblue} In this part, we present a major result of the paper---the proof of uniqueness of {\color{dblue} Matrix Partition Codes}. We will need to first illustrate how a parent matrix can be extracted from arbitrary compression. This in turn requires the following lemma. }

\begin{lem}
\label{lemma9}
Given a compression $(H_1,\cdots,H_s)$ of $\mathcal S$, we define an $(s-1)n \times sn$ matrix

\begin{equation}
X=
\begin{pmatrix}
I & -I & 0 & \cdots & 0 & 0 \\
0 & I  & -I & \cdots & 0 & 0 \\
     & & \cdots & \\
     0 & 0 &  0 & \cdots & I & -I
\end{pmatrix},
\label{eqn9.1}
\end{equation}
where $I$ denotes the $n \times n$ identity matrix, and an $M \times sn$ matrix  
\begin{equation}
J=
\begin{pmatrix}
H_1 & \cdots & 0 \\
     & \cdots \\
     0& \cdots & H_s
\end{pmatrix}.
 \label{eqn9.2}
\end{equation}
We have $(X,J)$ forms a compression for another {\color{dblue} source with {\color{mygreen} deviation symmetry}}
\begin{equation}
\mathcal S'=\{({\bf v},{\bf v}+{\bf d}') | {\bf v}\in \mathbb F^{sn}, {\bf d}'\in \tilde{\mathcal D}\} \subset \mathbb F^{sn} \times \mathbb F^{sn},      
\label{eqn9.3}
\end{equation}
where $\color{myblue}\tilde{\mathcal D}\color{black}$ was defined in \eqref{eqn2.1}.
If in addition, $(H_1,\cdots,H_s)$ is a perfect compression for $S$, then $(X, J)$ is a perfect
compression for $S'$.
\end{lem}
\begin{proof}
Suppose
\begin{equation}
\left( \left.
 X  
\begin{pmatrix}
 {\bf v}_1 \\ \vdots \\ {\bf v}_s 
\end{pmatrix} \right|
J 
\begin{pmatrix}
 {\bf v}_1+{\bf d}_1  \\ \vdots \\ {\bf v}_s+{\bf d}_s
\end{pmatrix}
\right)
=
\left( \left.
 X  
\begin{pmatrix}
 {\bf u}_1 \\ \vdots \\ {\bf u}_s 
\end{pmatrix} \right|
J 
\begin{pmatrix}
 {\bf u}_1+{\bf f}_1  \\ \vdots \\ {\bf u}_s+{\bf f}_s
\end{pmatrix}
\right)
\label{eqn9.4}
\end{equation}
where ${\bf v}_i, {\bf u}_i \in \mathbb F^n$ for all $i$; $({\bf d}_1,\cdots, {\bf d}_s), ({\bf f}_1,\cdots, {\bf f}_s)\in 
\mathcal D(\Leftrightarrow $
$
\begin{pmatrix}
 {\bf d}_1 \\ \vdots \\ {\bf d}_s
\end{pmatrix},
\begin{pmatrix}
 {\bf f}_1 \\ \vdots \\ {\bf f}_s
\end{pmatrix}
 \in \color{myblue}\tilde{\mathcal D}\color{black})$.
From the outputs of $X$, we get
\begin{equation}
\begin{pmatrix}
 {\bf v}_1 - {\bf v}_2 \\
{\bf v}_2 - {\bf v}_3 \\
\vdots \\
{\bf v}_{s-1} - {\bf v}_s
\end{pmatrix}
= 
\begin{pmatrix}
 {\bf u}_1 - {\bf u}_2 \\
{\bf u}_2 - {\bf u}_3 \\
 \vdots \\
{\bf u}_{s-1} - {\bf u}_s 
\end{pmatrix}
\Rightarrow
\begin{pmatrix}
 {\bf v}_1 - {\bf u}_1 \\
{\bf v}_2 - {\bf u}_2 \\
\vdots \\
{\bf v}_{s-1} - {\bf u}_{s-1}
\end{pmatrix}
= 
\begin{pmatrix}
 {\bf v}_2 - {\bf u}_2 \\
{\bf v}_3 - {\bf u}_3 \\
 \vdots \\
{\bf v}_s - {\bf u}_s 
\end{pmatrix}.
\label{eqn9.5}
\end{equation}
Thus we have
\begin{equation}
{\bf w}\triangleq {\bf v}_1-{\bf u}_1={\bf v}_2-{\bf u}_2=\cdots={\bf v}_s-{\bf u}_s. 
\label{eqn9.6}
\end{equation}
The outputs of $J$ give
\begin{align}
\begin{pmatrix}
 H_1({\bf v}_1+{\bf d}_1) \\ \vdots \\ H_s({\bf v}_s+{\bf d}_s)    
\end{pmatrix}
=
\begin{pmatrix}
         H_1({\bf u}_1+{\bf f}_1) \\ \vdots \\ H_s({\bf u}_s+{\bf f}_s)
\end{pmatrix} \\ \nonumber
\begin{pmatrix}
 H_1({\bf w}+{\bf d}_1) \\ \vdots \\ H_s({\bf w}+{\bf d}_s)          
\end{pmatrix}
= 
\begin{pmatrix}
 H_1({\bf f}_1) \\ \vdots \\ H_s({\bf f}_s)
\end{pmatrix}
\label{eqn9.7}
\end{align}

Since $({\bf w}+{\bf d}_1, \cdots , {\bf w}+{\bf d}_s)$ and  $({\bf f}_1,\cdots,{\bf f}_s)\in \mathcal S$ and 
$(H_1,\cdots,H_s)$ is a compression for $S$, we have
\color{myblue}
$({\bf w}+{\bf d}_1,\cdots, {\bf w}+{\bf d}_s)=({\bf f}_1,\cdots, {\bf f}_s)$. By \eqref{eqn1.11} and the fact 
that $({\bf d}_1,\cdots, {\bf d}_s)$ and $({\bf f}_1,\cdots, {\bf f}_s)$ are both in $\mathcal D$, we get
\color{black}
\begin{equation}
\bf w=0 \mbox{ and } 
({\bf d}_1,\cdots,{\bf d}_s)=({\bf f}_1,\cdots,{\bf f}_s), 
\label{eqn9.8}
\end{equation}
i.e. 
\begin{equation}
 \begin{pmatrix}
  {\bf v}_1 \\ \vdots \\ {\bf v}_s
 \end{pmatrix}=
\begin{pmatrix}
 {\bf u}_1 \\ \vdots \\ {\bf u}_s
\end{pmatrix}
\mbox{ and }
\begin{pmatrix}
 {\bf v}_1+{\bf d}_1 \\ \vdots \\ {\bf v}_s + {\bf d}_s
\end{pmatrix}=
\begin{pmatrix}
 {\bf u}_1 + {\bf f}_1 \\ \vdots \\ {\bf u}_s + {\bf f}_s
\end{pmatrix}.
\label{eqn9.9}
\end{equation}
Thus, $(X, J)$ is a compression for $\mathcal S'$. 
\color{myblue}

Finally, if $(H_1,\cdots, H_s)$ is a perfect compression, then $|\mathcal C|=|\mathcal S|=|\mathbb F|^n|\mathcal D|$. On the other hand, 
the target space of $(X,J)$ is $\mathbb F^{(s-1)n}\times \mathcal C$, whose cardinality is $|\mathbb F|^{(s-1)n}|\mathcal C|=|\mathbb F|^{sn}|\mathcal D|=|\mathbb F|^{sn}|\tilde{\mathcal D}|=|\mathcal S'|$. Therefore, $(X,J)$ is also a perfect compression.
\color{black} 

\end{proof}
 
\begin{thm}[Existence of Parent Matrix]
 \label{thm9}
Given a compression $(H_1,\cdots ,H_s)$, 
$\exists$ a surjective parent matrix $P$ of $\mathcal S$ 
satisfying \eqref{eqn2.2} and \eqref{eqn2.3} with $\mbox{null} H_i\subset \mbox{null} Q_i$ for all $i$.
Moreover, if $H_i$ are surjective for all $i$, then $P$ is an $(M-n) \times sn$
matrix \color{myblue} such that for finite $\mathbb F$, $P|_{\tilde{\mathcal D}}$ is bijective if and 
only if $(H_1,\cdots,H_s)$ is a
perfect compression. \color{black}
\end{thm}

\begin{proof}
Lemma \ref{lemma9} tells us that the {\color{dblue} corresponding $(X,J)$} defined in \eqref{eqn9.1} and
\eqref{eqn9.2} is a compression 
of $\mathcal S'$ (c.f. \eqref{eqn9.3}). By Theorem \ref{thm8.1}, any two matrices with null spaces \{0\} and 
$\mbox{null} X\oplus \mbox{null} J$ also forms a compression for $\mathcal S'$. Therefore,
$(I_{sn \times sn}, P)$ is a compression of $\mathcal S'$, where $P$ is an $r \times sn$ surjective matrix 
with $\mbox{null} P=\mbox{null} X\oplus \mbox{null} J$. It follows that 
\begin{equation}
r=sn-dim \mbox{null} X-dim \mbox{null} J.                                                      
\label{eqn9.10}
\end{equation}
To see $P$ is the parent matrix that we are looking for, we partition $P$ into
$P=[Q_1,\cdots,Q_s]$, where $Q_i$ are $r \times n$ matrices for all $i$. We have
$Q_1+Q_2+...+Q_s=0$ since 
\begin{equation}
\mbox{null} X=\left\{ \left.\begin{pmatrix}
                \bf v \\ \vdots \\ \bf v
               \end{pmatrix} \right|
\bf v\in \mathbb F^n \right\} \subset \mbox{null} P.
\label{eqn9.11}
\end{equation}
Thus $P$ {\color{black} satisfies} \eqref{eqn2.3}. 
Moreover,
\begin{equation}
\mbox{null} J=\left\{ 
\left.
\begin{pmatrix}
 {\bf n}_1 \\ \vdots \\ {\bf n}_s 
\end{pmatrix} \right| {\bf n}_i\in \mbox{null} H_i \mbox{ for } 1 \le i \le s \right\} \subset \mbox{null} P               
\label{eqn9.12}
\end{equation}
implies
\begin{equation}
\mbox{null} H_i \subset \mbox{null} Q_i \mbox{ for all } i.                                                  
\label{eqn9.13}
\end{equation}
To prove $P|_{\color{myblue}\tilde{\mathcal D}\color{black}}$ is injective (c.f. \eqref{eqn2.2}), we let ${\bf d}', {\bf f}' \in \color{myblue}\tilde{\mathcal D}\color{black}$. 
Suppose $P {\bf d}'=P {\bf f}'$. We have 
\begin{equation}
(I_{sn\times sn}{\bf 0}, P({\bf 0}+{\bf d}'))=(I_{sn \times sn}{\bf 0}, P({\bf 0}+{\bf f}')).                       
\label{eqn9.14}
\end{equation}
Since both $({\bf 0}, {\bf 0}+{\bf d}')$ and $({\bf 0}, {\bf 0}+{\bf f}')$ belong to $\mathcal S'$ and $(I_{sn \times sn}, P) |_{\mathcal S'}$ is
injective, we get ${\bf d}'={\bf f}'$ and hence $P$ satisfies \eqref{eqn2.2}.

Next if $H_i$ are surjective for all $i$, then
\begin{align}
M+dim\mbox{null}(H_1,\cdots,H_s) & = sn \\ \nonumber
                M+dim \mbox{null} J &= sn \\ \nonumber
                              M-n &  = sn - dim \mbox{null} J -dim \mbox{null} X       
\label{eqn9.15}
\end{align}
because $dim \mbox{null} X=n$. Hence \eqref{eqn9.10} becomes
\begin{equation}
                                    r=M-n                                                    
\label{eqn9.16}
\end{equation}
that $P$ is an $(M-n) \times sn$ matrix. \color{myblue}
Lastly for finite $\mathbb F$, $(H_1,\cdots, H_s)$ is a perfect compression iff 
$|\mathbb F|^{M-n}=|\tilde{\mathcal D}|$ by \eqref{eqn6.0}, iff $P|_{\tilde{\mathcal D}}$ is bijective as we have already shown $P|_{\tilde{\mathcal D}}$ is injective.
\color{black}
\end{proof}
%


\color{myblue}
\begin{thm}[Uniqueness of Partition Codes]
\label{thm10}
 Every linear lossless compression of a source with deviation symmetry is a Pre-Matrix Partition Code with a parent matrix obtained in Theorem 
\ref{thm9}. If the compression is linear-optimal or perfect, then the code is a Matrix 
Partition Code. 
\end{thm}

\color{black}
\begin{proof}
Let $(H_1,\cdots,H_s)$ be a compression and $P$ be the corresponding parent matrix in Theorem \ref{thm9}. We have 
$\mbox{null} \begin{pmatrix}
       Q_1 \\ \vdots \\ Q_s \\ H_1 \\ \vdots \\ H_s
      \end{pmatrix}
\subset
\mbox{null} 
\begin{pmatrix}
 H_1 \\ \vdots \\ H_s 
\end{pmatrix}
= \{{\bf 0}\}$ by Lemma \ref{lem8.1}. 
Hence 
$\begin{pmatrix}
 Q_1 \\ \vdots \\ Q_s \\ H_1 \\ \vdots \\ H_s 
\end{pmatrix}$
 is injective. By Theorem \ref{thm2.1}, $(H'_1,\cdots,H'_s)$ is a compression of Pre-Matrix 
Partition Code
if $\mbox{null} H'_i=\mbox{null} \begin{pmatrix}
                    Q_i \\ H_i
                   \end{pmatrix}$ for all $i$. In particular, $(H_1,\cdots,H_s)$ itself is such a 
compression because $\mbox{null} H_i\subset \mbox{null} Q_i$ (c.f. \eqref{eqn9.13}). 

\color{myblue}
If in addition that the compression code is linear-optimal or perfect, then it
 must be a Matrix Partition Code by the coronaries of Theorem \ref{thm2.2}.
\end{proof}
\color{black}

\color{myblue}
Given an $\mathcal S$, there always exists a linear-optimal compression $(H_1,\cdots,H_s)$ for it. 
\color{myblue} For finite field, \color{myblue} if the compression is perfect, then $|\mathcal S|=|\mathcal C|$ and we could not extend $\mathcal S$ without compromising the 
minimal compression {\color{mbrown}sum-ratio}. Otherwise, we have
room to add more elements to $\mathcal S$ without changing the compression. We can enlarge (see the proof of Theorem \ref{thm6}) the corresponding set $\color{myblue}\tilde{\mathcal D}\color{black}$ (and hence the $\mathcal S$ per se) until the surjective parent matrix $P$ (c.f. Theorem \ref{thm9}) we are working with becomes bijective when restricted to the extended $\tilde{\mathcal D}$. The compression for the extended $\mathcal S$ is now perfect and we cannot extend thing further. 
Hence the extended $\mathcal S$ is one of the largest sets containing $\mathcal S$ that admit the same minimal {\color{mbrown}sum-ratio}. 
Conversely, let $\mathcal S'\supset\mathcal S$ and both admit the same minimal {\color{mbrown}sum-ratio}. Let $(H'_1,\cdots,H'_s)$ be a linear-optimal compression for $\mathcal S'$.   
The compression must be perfect or $\mathcal S'$ is not one of those largest by the same argument. Notices that the compression works for $\mathcal S$, a subset of $\mathcal S'$. Actually, it is a linear-optimal compression for $\mathcal S$.

\color{myblue}
This method does not work for infinite field.
Even with both $(H_1,\cdots,H_s)|_{\mathcal S}$ and $P|_{\tilde{\mathcal D}}$
being bijective, there can be another compression $(H'_1,\cdots, H'_s)$ with
the same compression {\color{mbrown}sum-ratio} such that $(H'_1,\cdots,H'_s)|_{\mathcal S}$ 
and $P'|_{\tilde{\mathcal D}}$ is merely injective. Thus, $\mathcal S$ can still be extended without compromising the minimal {\color{mbrown}sum-ratio}. Here is an example.
Let $\mathbb F=\mathbb R$, $s=2$, $n=2$ and
\begin{equation}
\mathcal D = \{({\bf 0},a{\bf e}_1),({\bf 0},b{\bf e}_2)\mid |a|\ge1, |b|<1 \}.
\end{equation}
Notice that the two dimensional vector space $\{({\bf v},{\bf v}) | {\bf v}\in \mathbb R^2\}$ 
is a proper subset of $\mathcal S$.  Hence $\mathcal S$ cannot be 
compressed into $\mathcal C$ if dim$\hspace{0.1cm} \mathcal C\le 2$. It is not difficult to 
see that $\left(\begin{pmatrix}1 & 0\\ 0 & 1\end{pmatrix}, (1\hspace{0.3cm}1)\right)$ is a compression for $\mathcal S$ and
it is bijective when restricted to $\mathcal S$. The compression is
linear-optimal as dim$\hspace{0.1cm}\mathcal C=3$. The parent matrix 
$P=(-1\hspace{0.1cm} -\hspace{-0.1cm}1\hspace{0.1cm} 1 \hspace{0.1cm}1)$ is also 
bijective when restricted to the corresponding $\tilde{\mathcal D}$.
Now $\left(\begin{pmatrix}1 & 0\\ 0 & 1\end{pmatrix}, (2\hspace{0.3cm}1)\right)$ is another compression for $\mathcal S$ with 
parent matrix $P'= (-2\hspace{0.1cm} -\hspace{-0.1cm}1\hspace{0.1cm} 2 \hspace{0.1cm}1)$.
Neither $\left(\begin{pmatrix}1 & 0\\ 0 & 1\end{pmatrix}, (2\hspace{0.3cm}1)\right)|_{\mathcal S}$ nor $P'|_{\tilde{\mathcal D}}$ is bijective.
\color{black}

\color{black}

\section {Matrix Partition Codes for Hamming Sources}
\label{sect:hamming}

{\color{myblue} In this section, we will use {\color{dblue} Hamming sources described} in \eqref{eqn1.5} to give more concrete examples for {\color{dblue} Matrix Partition Codes}. Moreover, we will discuss linear-optimal compression for Hamming sources over both finite and infinite fields.

}

\subsection{Parent matrix $P$ of a {\color{dblue} Matrix Partition} Code for a Hamming Source}


Recall a Hamming source $\mathcal S$ described by \eqref{eqn1.5}.
For $s \ge 3$, we have 
$\mathcal D$ in \eqref{eqn1.6}. The corresponding 
\begin{equation}
\color{myblue}\tilde{\mathcal D}\color{black}=\{a {\bf e}_i | a\in \mathbb F, 1 \le i \le sn\} \subset \mathbb F^{sn}
\quad \mbox{(c.f. \eqref{eqn2.1}).} 
 \label{eqn3.1}
\end{equation}
To have an $r\times sn$ matrix $P$ satisfying \eqref{eqn2.2}, the necessary and sufficient condition is
\begin{equation}
\mbox{
each column of $P$ can't be the multiple of the other.                
 }
\label{eqn3.2}
\end{equation}
Say if ${\bf P}_i=a {\bf P}_j$,
where ${\bf P}_i$ and ${\bf P}_j$ are the i-th and j-th column of $P$, respectively; and
$a\in \mathbb F$. Then $P({\bf e}_i)=P(a{\bf e}_j)$ and 
\eqref{eqn2.2}
implies $i=j$ and $a=1$.
Conversely if $P(a {\bf e}_i)=P(b {\bf e}_j)$, then $a{\bf P}_i=b {\bf P}_j$. Condition \eqref{eqn3.2} will
imply $a=b=0$ or $a=b\neq 0$ and $i=j$, i.e., \eqref{eqn2.2} will be fulfilled.
As $sn \ge 3$, we have 
\begin{equation}
r>1.                                                                                    
 \label{eqn3.3}
\end{equation}

For infinite $\mathbb F$, we can always achieve \eqref{eqn3.2} with $r=2$. Explicitly, $P$ can be
\begin{equation}
P=
\begin{pmatrix}
 1 & 1 & \cdots & 1 \\
a_1 & a_2 & \cdots & a_{sn}
\end{pmatrix},
\label{eqn3.4}
\end{equation}
where $a_1, a_2,\cdots, a_{sn}$ are distinct. When $\mathbb F$ is finite. The condition 
\eqref{eqn3.2}
becomes
\begin{equation}
sn \le (|\mathbb F|^r-1)/(|\mathbb F|-1).                                                           
 \label{eqn3.5}
\end{equation}
Take $r=2$ and $\mathbb F=\mathbb Z_5$ as an example. Condition 
\eqref{eqn3.5}
 gives $sn \le 6$ and 
$P$ can be any segment of
\begin{equation}
\begin{pmatrix}
0 & 1 & 1 & 1 & 1 & 1 \\
1 & 0 & 1 & 2 & 3 & 4
\end{pmatrix}
\label{eqn3.6}
\end{equation}

Now we consider $s=2$ with $\mathcal D$ given by 
\eqref{eqn1.9}. The corresponding
\begin{equation}
\color{myblue}\tilde{\mathcal D}\color{black}=\{a{\bf e}_i | a\in \mathbb F, s<i \le 2s\}.
 \label{eqn3.7}
\end{equation}
Let $P=[Q_1 | Q_2]$ with $Q_1$ and $Q_2$ are $r \times n$ matrices. 
 The necessary and sufficient condition for $P$ to satisfy
$\eqref{eqn2.2}$ become
\begin{equation}
\mbox{
$Q_1$ is arbitrary;  each column of $Q_2$ can't be the multiple of the other. 
}
\label{eqn3.8}
\end{equation} 
Obviously, we can further set $Q_1=-Q_2$ such that \eqref{eqn2.3} will be satisfied.
Again we can always achieve 
\eqref{eqn3.8} 
with $r=2$ if $\mathbb F$ is infinite, e.g. we can set
\begin{equation}
Q_2=
\begin{pmatrix}
 1 & 1 & \cdots & 1 \\
a_1 & a_2 & \cdots & a_n
\end{pmatrix},
\label{eqn3.9}
\end{equation}
with distinct elements $a_1,\cdots,a_n$ of $\mathbb F$.
For finite $\mathbb F$, the condition 
\eqref{eqn3.8}
become
\begin{equation}
n \le (|\mathbb F|^r-1)/(|\mathbb F|-1).                                                                                
\label{eqn3.10}
\end{equation}

\subsection{Linear-Optimal compression for Hamming Source with $s=2$}


As in the last section 
\eqref{eqn3.7}-\eqref{eqn3.10}
mentioned. We take
\begin{equation}
P= 
[\mbox{
\begin{tabular}{ @{} c @{} ;{2pt/2pt} @{} c @{} }
 $\hspace{-0.1cm} -Q_2 \hspace{0.03cm}$ &  $\hspace{0.08cm} Q_2 \hspace{-0.1cm}$
\end{tabular}
}].                                                                                    
\label{eqn4.1}
\end{equation}
with $Q_2$ satisfying the condition \eqref{eqn3.8}. \color{myblue} Let $C_2$ be 
a row basis matrix of $Q_2$ as in the Matrix Partition Code. Then $C_2$ also fulfills \eqref{eqn3.8}. For simplicity, we assume $Q_2$ itself is surjective and $C_2=Q_2$. 
Notice that $Q_2$ also equal to the $Y$ in \eqref{eqn2.18g}
Let $T=\begin{pmatrix}
 G_1 \\ G_2
\end{pmatrix}$ be the matrix that fulfills \eqref{eqn2.18z}, we have
\begin{equation}
\left(
U_1\begin{pmatrix}
 G_1 \\Q_2
\end{pmatrix},
U_2\begin{pmatrix}
 G_2 \\Q_2
\end{pmatrix}
\right)
\mbox{ forms a Matrix Partition Code for $\mathcal S$, }
\label{eqn4.4} 
\end{equation}
for any invertible matrices $U_1, U_2$ with appropriate sizes.
Notice that $\begin{pmatrix}
 T \\ Q_2
\end{pmatrix}$
is invertible by \eqref{eqn2.18z} with $Q_2=Y$. Let $U_1$ be the inverse of  $\begin{pmatrix}
 T \\ Q_2
\end{pmatrix}$ and $U_2$ be an identity matrix.
Put $G_1=T$ and $G_2$ as void, we get a compression
\begin{equation}
(I_{n \times n}, Q_2).
\label{eqn4.4b}
\end{equation} 
 Let $Q_2$ be an $r'\times n$ matrix. The total code length for \eqref{eqn4.4b} and hence for \eqref{eqn4.4} as well is
\begin{equation}
M=n+r'
\label{eqn4.5}
\end{equation}
\color{black}

For finite $\mathbb F$ and a given $n$.
We must have $|\mathcal S| \le |\mathcal C|$, ie
\begin{align}
|\mathcal S|=|\mathbb F|^n(1+n(|\mathbb F|-1)) \le |\mathbb F|^{n+r'} 
 \nonumber \\
n \le (|\mathbb F|^{r'}-1)/(|\mathbb F|-1).                                                              
\label{eqn4.6}
\end{align}
The total code length $M$ will be minimized if $r'$ is the smallest 
integer satisfying \eqref{eqn4.6}. Fortunately, such $r'$ observes the sufficient condition
\eqref{eqn3.10} with $r=r'$. That means we can always pick a $Q_2$ with that $r$ and get
linear-optimal compressions.

For infinite $\mathbb F$ and $n=1$, $\{[1], [1]\}$ is a linear-optimal compression, since  something like $\{[1], void\}$ would give the same output 
for $\{{\bf e}_1,{\bf e}_1\}$ and $\{{\bf e}_1, 2{\bf e}_1\}$. 
For $n\ge 2$, we can always pick $r=2$ with $Q_2$ defined in \eqref{eqn3.9} to get  Matrix Partition Codes \eqref{eqn4.4} with $M=n+2$. 
 The compression obtained is actually linear-optimal. Let $(H'_1,H'_2)$ be another compression 
with another codeword space $\mathcal C'$. The proper subset
\begin{equation}
\mathcal B=\{({\bf v},{\bf v}+a{\bf e}_1) | {\bf v}\in \mathbb F^n, a\in\mathbb F \} \subset \mathcal S 
\label{eqn4.7}
 \end{equation}
is a vector space with dimension $n+1$. Since all of the compressions considered are linear and
injective within $\mathcal S$, the output of $\mathcal B$ is a vector subspace
of $\mathcal C'$ with dim $n+1$. To accommodate the output of the $({\bf 0},{\bf e}_2)$,
which belongs to $\mathcal S$ but not $\mathcal B$, {\color{black} $dim \mathcal C'$} must be at
least $n+2$.

\subsection{Optimal Lossless Compression for Hamming Source over Infinite Fields}


We have found such a compression for $s=2$ in the  previous section. So let $s
\ge 3$. For $n=1$, the linear-optimal compression is $\{[1],\cdots,[1]\}$. The reason
is the same as the case of $s=2$. 

Now let $s \ge 3$ and $n \ge 2$. 
First we will deal with the compression {\color{dblue} ratio}.  By more or less the same argument
of the two-source case (c.f. \eqref{eqn4.7}), the dimension of the codeword space
$\mathcal C$ cannot be less than $n+2$.
In addition, each column within any encoding matrix  cannot be multiple of each
other. Say if $(H_1)_i=a(H_1)_j$, then {\color{black} $(a {\bf e}_i, {\bf 0},\cdots,{\bf 0})$ and
$({\bf e}_j,{\bf 0},\cdots,{\bf 0})$} will share the same output. That means each encoding matrix
has at least $2$ rows since $n \ge 2$. Therefore we have
\begin{equation}
dim \mathcal C \ge \max(n+2, 2s).      
\label{eqn5.1}
\end{equation}
Now we are going to build a compression for $S$ over $\mathbb F$ with characteristic 
0 and $dim \mathcal C =
\max(n+2,2s)$. By \eqref{eqn5.1}, that compression will be a linear-optimal compression.
We construct a $2 \times sn$ matrix $P$, the parent matrix, through its component
$P=[Q_1|\cdots|Q_s]$ (c.f. \eqref{eqn2.3}). We will make use of the fact that $\mathbb 
F$ contains all rational numbers as a subfield. 
Let $p_1, p_2, \cdots,p_s, q$ be prime numbers such that
\begin{equation}
0<p_s=p_1<p_2<\cdots<p_{s-1}<q.      
\label{eqn5.2} 
\end{equation}
Let 
\begin{align}
\left\{ 
\begin{matrix}
t_{2i-1}&=&p_iq  \\
      t_{2i}&=&p_{i+1}
      \end{matrix}
      \right.
\qquad \mbox{ for $1 \le i<s$}.
\label{eqn5.3a} 
\end{align}
We define the $j$-th column of $Q_i$ as
\begin{equation}
(Q_i)_j= 
\begin{pmatrix} 
t_{2i-1}^j \\
          t_{2i}^j
         \end{pmatrix}
=
\begin{pmatrix} 
 (p_iq)^j \\
         p_{i+1}^j
         \end{pmatrix}
              \mbox {for $1\le i<s$, $1 \le j \le n $.}                                          
\label{eqn5.3b}
 \end{equation}
They are not multiplier of each other. Say if $(Q_i)_j$ is a multiplier of $(Q_k)_l$, then
\begin{equation}
(p_iq)^jp_{k+1}^l=(p_kq)^lp_{i+1}^j,                                                       
\label{eqn5.4}
\end{equation}
which gives $i=k$ and $l=j$.
We put $Q_s=-(Q_1+\cdots+Q_{s-1})$ and get
\begin{equation}
(Q_s)_j=-\begin{pmatrix}
q^j(p_1^j+p_2^j+\cdots+p_{s-1}^j) \\
                p_2^j+p_3^j+\cdots+p_s^j          
         \end{pmatrix}
\mbox{ for $1 \le j \le n$},                      
\label{eqn5.5}
 \end{equation}
which is a multiplier of 
$\begin{pmatrix}
  q^j \\
       1
 \end{pmatrix}$
 since $p_1=p_s$.                                                                  
  Obviously the vectors in 
\eqref{eqn5.5} are not multipliers of the others nor  multipliers of those in
\eqref{eqn5.3b}. Hence by \eqref{eqn3.2}, our $P$ satisfies \eqref{eqn2.2}. It
fulfils \eqref{eqn2.3} by construction.

If $2(s-1) \ge n$, we let $T'$ (c.f. \eqref{eqn2.5}) be a void matrix and have
an injective $2(s-1) \times n$ matrix
\begin{equation}
R=\begin{pmatrix}
Q_1 \\
      \vdots \\
      Q_{s-1}
\end{pmatrix}=
{\color{nblue}
\begin{pmatrix}
 t_1 & t_1^2 & \cdots  & t_1^n \\
 t_2 & t_2^2 & \cdots & t_2^n \\
 & & \cdots & \\
 t_{2(s-1)} & t_{2(s-1)}^2 & \cdots &  t_{2(s-1)}^n                                                                         
\end{pmatrix}}
\label{eqn5.7} 
\end{equation}
It is injective as it contains the minor 
{\color{nblue}
\begin{equation}
\begin{pmatrix}
 t_1 & t_1^2 & \cdots  & t_1^n \\
 t_2 & t_2^2 & \cdots & t_2^n \\
 & & \cdots & \\
 t_n & t_n^2 & \cdots &  t_n^n                                                                         
\end{pmatrix}
\label{eqn5.8}
\end{equation}}
whose determinant is 
{\color{nblue}
\begin{equation}
t_1t_2 \cdots t_n \prod_{i>j}(t_i-t_j) \neq 0 \mbox{ (c.f. \eqref{eqn5.3a})}.    
\label{eqn5.9}
 \end{equation}}
By Theorem {\color{black} \ref{thm2.1}},
$(Q_1,\cdots,Q_s)$ is a compression for $\mathcal S$, a linear-optimal  compression (c.f. \eqref{eqn5.1}). It is
also a {\color{dblue} Matrix Partition Code} with void $T$ and $C_i=Q_i$ for all $i$.
 
For $2(s-1)<n$, we pick any nonzero {\color{nblue}$t_{2s-1},\cdots,t_n$}  such that {\color{nblue} $t_i \neq t_j$} for
all $i \neq j$. We let 
{\color{nblue}
\begin{equation}
T'=
\begin{pmatrix}
t_{2s-1} & t_{2s-1}^2 & \cdots & t_{2s-1}^n \\
     & & \cdots & \\
      t_n      &  t_n^2 & \cdots & t_n^n                                                          
\end{pmatrix}.
\label{eqn5.10}
\end{equation}}
Then 
{\color{nblue}
\begin{equation}
R=\begin{pmatrix}
Q_1 \\
     \vdots \\
     Q_{s-1} \\
      T'   
  \end{pmatrix}
= \begin{pmatrix}
 t_1 & t_1^2 & \cdots  & t_1^n \\
 t_2 & t_2^2 & \cdots & t_2^n \\
 & & \cdots & \\
 t_n & t_n^2 & \cdots &  t_n^n                                                                         
\end{pmatrix},
\label{eqn5.11}                                                                                            \end{equation}}
whose determinant is not zero as all $t_i$ are distinct. Hence $R$ is injective
(actually bijective) and by Theorem \ref{thm2.1}, 
$\begin{pmatrix}
 G_1 \\ Q_1 
 \end{pmatrix}, \cdots, 
\begin{pmatrix}
 G_s \\ Q_s
\end{pmatrix}
$
is a compression of $S$, a linear-optimal compression again, where            
$ \begin{pmatrix}
   G_1 \\ \vdots \\ G_s 
  \end{pmatrix}
 = T'$.               
It is a {\color{dblue} Matrix Partition Code} with $T=T'$, 
$Y=\begin{pmatrix}
    Q_1 \\ \vdots \\ Q_{s-1}
   \end{pmatrix}
$
and $C_i = Q_i$ for all $i$. 

Actually, for infinite field $\mathbb F$ with any characteristic, the chance for two randomly 
picked vectors (with more than one entry) to be multipliers of the others are virtually 
zero. So one may just pick $Q_1,Q_2,\cdots,Q_{s-1}$ randomly, instead of \eqref{eqn5.3b}.
Then define $Q_s$ as $-(Q_1+Q_2+\cdots +Q_{s-1})$. The chance to have a pair of columns 
in $P=[Q_1|\cdots|Q_s]$ that are multipliers of the others are again virtually zero. 
Moreover,  
$ \begin{pmatrix}
   Q_1 \\ \vdots \\ Q_{s-1} 
  \end{pmatrix}
$ should have full rank as its entries are randomly selected. So if $2(s-1)\ge n$, then 
$ \begin{pmatrix}
   Q_1 \\ \vdots \\ Q_{s-1} 
  \end{pmatrix}$ should be injective and $(Q_1,\cdots,Q_s)$ forms a linear-optimal compression 
like the aforementioned example. For $2(s-1)<n$, we should be able to augment 
$\begin{pmatrix}
   Q_1 \\ \vdots \\ Q_{s-1} 
  \end{pmatrix}$
to a bijective 
$\begin{pmatrix}
   Q_1 \\ \vdots \\ Q_{s-1}\\ T 
  \end{pmatrix}$. Then 
$\begin{pmatrix}
 G_1 \\ Q_1 
 \end{pmatrix}, \cdots, 
\begin{pmatrix}
 G_s \\ Q_s
\end{pmatrix}
$ will forms an linear-optimal compression for $S$ as before.

\color{myblue}         
\section{Structure of Deviation Symmetry}
\label{sect:fixing}
Given $n,s,\mathbb F$, let $\Sigma$ be the set of sources with deviation symmetry. For any
$\mathcal S\in\Sigma$, we pick a representative set $\mathcal D(\mathcal S)$ of it. Notice that
we have $n$ degree of freedom in choosing a particular {\color{mbrown}$\delta\in\mathcal D(\mathcal S)$}.
Therefore it is possible to fix a certain component to be a constant within the whole $\mathcal D(\mathcal S)$. Say if we want to fix the last component to be zero, then we simply replace
${\color{mbrown}\delta=}({\bf d}_1,\cdots, {\bf d}_s)$ with $({\bf d}_1,\cdots, {\bf d}_s)-({\bf d}_s,\cdots, {\bf d}_s), \forall ({\bf d}_1,\cdots,{\bf d}_s)\in \mathcal D(\mathcal S)$. We call such fixing as component-fixing. Since each component contains $n$ entries, the component-fixing eliminates all $n$ degrees of freedom in choosing {\color{mbrown}$\delta$}. 

Let us impose a component-fixing throughout the $\Sigma$. It can be shown that $\mathcal S_1\subset \mathcal S_2$ if and only if 
$\mathcal D(\mathcal S_1)\subset\mathcal D(\mathcal S_2)$. Now $\mathcal D$ can be viewed as an injective 
mapping. It is more than that. More specifically, after fixing the last component to be zero, then $\mathcal D: \Sigma\rightarrow \mbox{power set of }\underbrace{{\color{mbrown}\mathbb F^n\times \cdots \times \mathbb F^n}}_{s-1\mbox{ terms}}\times\{\bf 0\}$ becomes a bijective mapping. Power set of a set is a $\sigma$-algebra for sure. Actually, we can show by definition (2.1) that $\Sigma$ is also a $\sigma$-algebra. The bijective
mapping $\mathcal D$ preserves their structures in sense that $\mathcal D(\bigcup_{i=1}^\infty \mathcal S_i)=\bigcup_{i=1}^\infty \mathcal D(\mathcal S_i)$, $\mathcal D(\mathcal S^c)=\mathcal D(\mathcal S)^c$ and $\mathcal D(\emptyset)=\emptyset$, where $c$ stands for complement. 

Throughout the paper, we say $\mathcal S$ can be compressed by the encoding matrices $(H_1,\cdots,H_s)$ if $(H_1,\cdots,H_s)|_\mathcal S$ is injective, which does not specify if the source actually is compressed into a smaller space. To distinguish thing, we will say $S$ is compressible if and only if it can be compressed by
$(H_1,\cdots,H_s)$ into a lower dimensional space  (i.e. dim $\mathcal C<sn$). We will make use of the component-fixing  to
determine the necessary and sufficient condition for $\mathcal S$ to be compressible.

\begin{thm}
\label{thm4s1}
$\mathcal S$ is compressible if and only if there exists a subset 
$\mathcal D_-$ of $\mathbb F^{n-1}\times\underbrace{\mathbb F^n\times\cdots\times\mathbb F^n}_{s-2\mbox{ terms}}$ such that the representative set $\mathcal D$ can be expressed as
\begin{equation}
\mathcal D=\left\{\pi\left(B\begin{bmatrix}f({\bf d}_1,\cdots, {\bf d}_{s-1})\\{\bf d}_1\end{bmatrix}, {\bf d}_2,..., {\bf d}_{s-1}, {\bf 0}\right)  | ({\bf d}_1,\cdots, {\bf d}_{s-1})\in \mathcal D_-\right\}, 
\label{eqn4s.1}
\end{equation}
where $\pi$ is a position permutation of the $\mathbb F^n$-vectors, $B$ is an $n\times n$ invertible matrix and $f$ is a well-defined function from $\mathcal D_-$ to $\mathbb F$.
\end{thm}
\begin{proof}
"$\Rightarrow$" Suppose $\mathcal S$ can be compressed by $(H_1,\cdots, H_s)$ into a lower dimensional space. One of the $H_i$, say $i=1$ WLOG, must have a nonzero null space. So let
${\bf 0}\neq {\bf v}\in \mbox{null} H_1$. Also let $A$ be an $(n-1)\times n$ matrix with null $A=$ span {$\{{\bf v}\}$}. 
Then $S$ can also be compressed by $(A,\underbrace{I_n,\cdots, I_n}_{s-1\mbox{ terms}})$ whose null space is a subspace of $(H_1,\cdots, H_s)$'s. Fix the last component of every  ${\color{mbrown}\delta} \in\mathcal D$ to be zero. For any $\sigma\in \mathcal S$, $\exists$ a unique ${\bf v}\in\mathbb F^n$ and a unique $({\bf d},{\bf d}_2,...,{\bf d}_s,{\bf 0})\in \mathcal D$ such that $\sigma=({\bf v}+{\bf d}, {\bf v}+{\bf d}_2,\cdots, {\bf v}+{\bf d}_{s-1}, {\bf v})$.
We have  
\begin{equation}
(A,\underbrace{I_n,\cdots, I_n}_{s-1\mbox{ terms}})\sigma=([{\bf 0} | I_{n-1}]B^{-1}({\bf v}+{\bf d}),{\bf v}+{\bf d}_2,\cdots,{\bf v}+{\bf d}_{s-1},{\bf v}),
\end{equation}
where $B$ is an $n\times n$ invertible matrix such that $AB=[{\bf 0} | I_{n-1}]$.
Therefore, we get everything back directly, except for the first entry of $B^{-1}{\bf d}$. Since the $\mathcal S$ can be compressed by $(A,\underbrace{I_n,\cdots, I_n}_{s-1\mbox{ terms}})$ (losslessly), we must be able to retrieve the lost part from the output. Mathematically, the first entry of $B^{-1} {\bf d}$ has to be a function of $[{\bf 0} | I_{n-1}]B^{-1}{\bf d}, {\bf d}_2,\cdots,{\bf d}_{s-1}$ and ${\bf v}$. However, we are talking about deviation symmetry that $\mathcal D$ does not depend on ${\bf v}$. Therefore $\mathcal D$ has the form of \eqref{eqn4s.1} with $\pi=1$ and ${\bf d}_1=[{\bf 0} | I_{n-1}]B^{-1} {\bf d}$.

"$\Leftarrow$" Conversely, given \eqref{eqn4s.1}, we let $H_1=[{\bf 0} | I_{n-1}]B^{-1}$, $H_2=H_3=\cdots = H_s = I_n$. Then $\mathcal S$ can be compressed by $(H_{\pi(1)},\cdots, H_{\pi(s)})$ losslessly.
\end{proof}

The argument in the theorem can be generalized until null $H_i\ne\{{\bf 0}\}$ for all $i$, i.e.
 actual compression happens at each terminal. 

\color{black}

\section{Conclusion}

{\color{myblue} In this paper, we study zero-error \color{myblue} linear \color{black} coding of a set of rather general sources known as {\color{dblue} sources with {\color{mygreen} deviation symmetry}}. {\color{dblue} Matrix Partition Codes} can be used to efficiently compress {\color{dblue} sources with {\color{mygreen} deviation symmetry}}. We will conclude here by summarizing the construction procedure of a {\color{dblue} Matrix Partition Code} in the following. }
Suppose we want to compress a {\color{dblue} source with {\color{mygreen} deviation symmetry}} $\mathcal S$  {\color{dblue} losslessly}. 
We can simply search the compression within the framework of Theorem \ref{thm2.1}
because Theorem \ref{thm10} tells us that there is no other way causing difference. So we 
need to fix a $\mathcal D$ (c.f. \eqref{eqn1.1}, \eqref{eqn1.3}) first. Theorem \ref{thm2.3} ensures that the choice of $\mathcal D$ does not \color{myblue} affect the end results. \color{black}  

Then we have to find the parent matrix $P$, an $r \times sn$ matrix satisfying
\eqref{eqn2.2} and \eqref{eqn2.3}. Such $P$ always exists as compression always exists. Precisely
all encoding matrices are identity matrix forms a trivial compression and $X$ (c.f. \eqref{eqn9.1})
is the corresponding parent matrix. The problem is about the compression
rate. Basically, $P$ with lower $r$ ends up with {\color{mbrown}more efficient compression}. On the other hand, it is easier to form the $P$ with higher $r$.
Once we get $P$, we can follow the mechanism of {\color{dblue} Matrix Partition Code} to get 
\color{myblue}a {\color{mbrown}code}
of highest compression efficiency (with that P) in sense of  Theorem \ref{thm2.2} and its second corollary.
\color{black}
{\color{nblue}
\section*{Acknowledgment}

We would like to thank the associate editor and the anonymous reviewers for their times and constructive comments.
}

\bibliographystyle{IEEEtran}
\bibliography{ref}

\end{document}